\documentclass[11pt,a4paper]{article}

\usepackage[margin=1in]{geometry}
\usepackage{amsmath}
\usepackage{amssymb}
\usepackage{amsthm}
\usepackage{array}
\usepackage{booktabs}
\newcolumntype{L}[1]{>{\raggedright\arraybackslash}p{#1}}
\usepackage{tikz}
\usepackage[numbers]{natbib}
\usepackage{needspace}
\usepackage[hidelinks,unicode]{hyperref}
\newcommand{\SourceDigest}{94babd15d7517d67d45053467f5a1b0b0f1abbc1fac62be5d8998d00d54298cf}
{unrecorded}}
\hypersetup{
  pdftitle={Finitely Verifiable Segmentation: Decidability for Every Tokenizer with a Finite-State Verifier},
  pdfauthor={Nicklas Nidhögg},
  pdfsubject={Segmentation policies; finite verifiers; decidability; rational relations},
  pdfkeywords={segmentation, tokenization, finite verifier, decidability, rational relations, source-sha256:\SourceDigest}
}
\theoremstyle{definition}
\newtheorem{definition}{Definition}

\theoremstyle{plain}
\newtheorem{theorem}{Theorem}
\newtheorem{lemma}{Lemma}
\newtheorem{corollary}{Corollary}
\newtheorem{proposition}{Proposition}

\newcommand{\code}[1]{\texttt{#1}}

\title{Finitely Verifiable Segmentation:\\Decidability for Every Tokenizer with a Finite-State Verifier}
\author{Nicklas Nidh\"ogg\\[0.4em] \normalsize Independent Researcher\\ \normalsize
\texttt{nicklas.nidhogg@gmail.com}\\ \normalsize ORCID: 0009-0006-6161-6150}
\date{October 16, 2026}

\begin{document}
\maketitle

\begin{abstract}
A segmentation policy assigns to each byte string at most one segmentation into tokens. Call it
\emph{finitely verifiable} when a finite automaton over marked strings, bytes paired with boundary bits, accepts exactly
the policy's own segmentations. The automaton judges a proposed segmentation and need not compute one. From a supplied
verifier the paper gives constructive decisions for four questions, whether a window certifies a token boundary, whether
the boundary gap is bounded, whether a mask can keep generation from dead-ending, and whether two policies are
equivalent. The last three specialize limitedness, co-accessibility and the equivalence of single-valued transducers.
The first refines the synchronizing cut of code theory, and two more ask whether a certified cut, or the cuts of an
inventory, let each piece scan alone. The class sits strictly between the online-scannable policies and those with
regular domains, and online scannability is decidable in polynomial time through Choffrut's theorem. Markings are
length-preserving, so the class is exactly the rational segmentation policies, a normal form rather than a restriction.
The class is closed under four pipeline operations, composition over a labeled stage, restriction to a regular schema,
splitting at character-class changes, and byte-fallback union of literal policies. Each composite inherits the six
decisions, and certificates transfer between stages under four laws. Under rational, unifilar, almost surely stopping
finite-state sources, each mass of the weighted section is an exact rational, and so is each rate where expected length
is positive. Maximal-munch lexing, canonical byte-pair encoding and two-mode WordPiece instantiate the interface, and
the munch library implements all six decisions. Every theorem carries an executable check, and every reported statistic
is a program's own output.
\end{abstract}

\section{Introduction}

A production system may need to know where its byte stream resynchronizes after a lost fragment, how far that stream can
run before a token boundary is guaranteed, whether a grammar-masked generator can be left with no admissible next token,
and what a tokenizer upgrade does to cached segmentations. Exact procedures answer some of these in particular settings,
such as masks in constrained decoding~\citep{koo2024automata} and the equivalence of byte-pair
dictionaries~\citep{berglund2024bpedfa}. Each question quantifies over all inputs, and a finite sample of inputs cannot
establish such a universal guarantee. From a supplied verifier, this paper gives constructive decisions for all four,
certification, the boundary gap, realizability and equivalence.

The series began with certified bytes~\citep{nidhogg2026splitpoints} and certified
windows~\citep{nidhogg2026splitwindows} for maximal-munch scanning, and its two companion studies answered special cases
of those questions. The recovery companion \citep{nidhogg2026panic} proved that certificates make resuming after damage
sound. The splitting companion \citep{nidhogg2026splitting} used the same certificates to divide work without
speculation, and decided certification, the anchor-free span, and segmentation differentials for regular token sets. The
present paper states what those decisions were always instances of. It rests on one hypothesis.

\Needspace{7\baselineskip}
\begin{definition}[Segmentation policy; finite verifiability]\label{def:fv}
A \emph{segmentation policy} assigns to each string at most one segmentation, none when the string is not tokenizable. A
segmentation here is an unlabeled partition of the string into consecutive nonempty pieces, equivalently its boundary
bit pattern, the boundary projection of what the splitting companion calls a segmentation. The names a stage gives those
pieces are separate labeler data, which is why two policies naming identical spans differently are the same policy. A
\emph{marked string} pairs each byte with one bit, whether a token boundary follows that byte's position, and a position
$i$ of it is a \emph{token start} exactly when $i = 0$ or the bit at $i-1$ is set. A policy is
\emph{finitely verifiable} when some deterministic finite automaton over marked strings accepts exactly the markings
that are the policy's own segmentations of its tokenizable strings.
\end{definition}

The verifier decides a \emph{proposed} segmentation by reading it. It need not compute the segmentation online, and that
distinction is the crux of the paper. Section~\ref{sec:real} exhibits a policy with no finite-state bounded-delay
scanner that is nevertheless finitely verifiable. Verification is the strictly weaker demand, and everything below asks
only for it.

The class is not new, and the paper does not present it as new. A finitely verifiable policy is a length-preserving
partial \emph{automatic function} in the sense of \citet{case2013automatic}, its graph a functional synchronous
relation, regular under convolution. The text keeps the verifier-first name because the verifier is the interface every
construction below programs against. The first three contributions below were not found in the transducer literature
that Section~\ref{sec:related} sets out, and are what the paper claims as its own; the last two locate the class among
known classes and apply it to real tokenizers.

The contributions are:

\begin{itemize}
\item The four decisions, and the two that ask whether certified cuts let the pieces scan alone, run on supplied
verifier presentations with constructive outputs and survive the four pipeline operations of Theorem~\ref{thm:calculus},
composition over a labeled stage, restriction to a regular schema, splitting at character-class changes and
byte-fallback union of literal policies (Sections~\ref{sec:four} and~\ref{sec:calculus}).
\item Certificates transfer between stages under proved laws, with an exact stage-local lift
(Section~\ref{sec:transfer}).
\item Zero differential mass is identified with the absence of a boundary differential in the source's support, and the
exact masses behind it are solved from the supplied verifiers through the same interface (Section~\ref{sec:weighted}).
\item The class sits strictly between the online-scannable policies and those with regular domains, its inner frontier
is decidable in polynomial time, and over the marking encoding it is exactly the rational segmentation policies
(Section~\ref{sec:real}).
\item Literal and regular maximal munch, canonical byte-pair encoding and two-mode WordPiece carry the interface, each
validated against an independent reference scanner or encoder, and the munch library implements all six decisions
(Sections~\ref{sec:interface} and~\ref{sec:instances}).
\end{itemize}

\begin{table}[tp]
\centering \small
\begin{tabular}{L{2.5cm}L{8.8cm}L{3.3cm}}
\toprule Result & Statement & Where \\ \midrule
Four decisions & Certification, the boundary gap and realizability are decidable for every supplied verifier, and full
equivalence for every pair. Certification returns a counterexample when negative, equivalence a differing string, the
boundary gap its exact supremum when bounded and a pumpable lasso otherwise, and realizability the mask. &
Theorem~\ref{thm:four} \\
\addlinespace Cut and chunk independence & Whether a certified cut lets both its sides scan alone, and whether all of
an inventory's cuts do at once, are decidable, each with a shortest counterexample. Cut independence of every pair
implies chunk independence, not conversely. Every maximal-munch cut and inventory is independent, and WordPiece and
threshold-splice ones need not be. & Theorems~\ref{thm:cut} and~\ref{thm:chunk} \\
\addlinespace Where the class sits & A regular token set with unbounded commit delay has a finite verifier but no
finite-state scanner. A non-regular domain puts a policy outside the class, and a regular domain does not put it
inside. & Section~\ref{sec:real}, Lemma~\ref{lem:outside}, Proposition~\ref{prop:regular-domain} \\
\addlinespace A normal form & Markings are length-preserving, so finite verifiability coincides with rationality of
the marking relation. Undecidability sits in presentations that do not convert effectively to a verifier. &
Proposition~\ref{prop:collapse}, Section~\ref{sec:real} \\
\addlinespace Closure & Composition over a labeled stage, restriction to a regular schema, splitting at character-class
changes and byte-fallback union of literal policies keep a policy in the class, and every composite inherits all six
decisions. & Theorem~\ref{thm:calculus} \\
\addlinespace Transfer & Restriction preserves certificates, and a splitter forces them at class changes. A pipeline's
token starts are a subset of its inner stage's. A certificate lifts through a pipeline exactly when a product of two
stage-local automata is empty, with a per-label sufficient condition as a corollary. A tower's certificate factors
through its levels on the tower's own domain. & Propositions~\ref{prop:monotone} to~\ref{prop:subset},
Theorem~\ref{thm:lift}, Corollary~\ref{cor:labels}, Proposition~\ref{prop:stack} \\
\addlinespace Weighted theory & Under rational, unifilar, almost surely stopping sources the tokenizability and
differential masses of Section~\ref{sec:weighted} are exact rationals, as are the unnormalized expected length and
expected token, mismatch and boundary counts, and every rate whose expected length is positive. Differential mass is
nonzero exactly when a boundary differential lies in the source's support. & Theorem~\ref{thm:weighted},
Proposition~\ref{prop:positivity} \\
\addlinespace Instances & Literal and regular-token maximal munch, canonical byte-pair encoding and two-mode WordPiece
carry the whole interface, each validated against an independent reference scanner or encoder. Under maximal munch a
certified window makes a cache keyed on its cut sound for continuations that keep the window and that the policy
tokenizes completely. & Section~\ref{sec:interface}, Proposition~\ref{prop:cache} \\ \bottomrule
\end{tabular}
\caption{Results under the hypothesis of Definition~\ref{def:fv}. Source comments or the gate name each executable
check, and the outside lemma and the collapse proposition are proofs alone.} \label{tab:results}
\end{table}

Table~\ref{tab:results} summarizes the results and names the statements that carry them. Three points of scope hold
throughout and are stated here once. First, the theorems assume a \emph{supplied} verifier, presented as its states,
transitions, acceptance and alphabet. A mere promise that some program computes a policy in the class is not enough,
since the reductions of Section~\ref{sec:real} leave the four questions without a uniform procedure. Second, membership
in the class is determined by the marking relation, whatever its presentation. A policy whose domain is not regular lies
outside the class, but a regular domain alone does not put it inside. Third, every validation sweep is bounded, in that
it checks every case within its declared universe and says nothing beyond it, so a sweep supports a theorem but never
replaces its proof.

\section{The definition and the interface}\label{sec:interface}


Every decision and construction here uses four methods, \code{start}, \code{step(state, byte, mark)} returning a
successor state or refusal, \code{accepts(state)}, and \code{alphabet}. The theorems cover an object with those four
methods when it presents a \emph{finite, effective} verifier, which means finitely many reachable states and a finite
alphabet, both enumerated once per decision. The executable modules take the alphabet's symbols as one-character
strings, and windows and the strings a decision returns are spelled over them. The verifier's accepted language must be
\emph{functional}, at most one accepted marking per byte string, and \emph{well formed}, every nonempty accepted marking
ending in a zero bit, the final-zero convention, so that it encodes a segmentation. A byte outside an object's own
alphabet leaves its domain by definition, and every decision treats that side as dead without stepping its verifier, so
the interface requires no behavior of a verifier on a foreign byte. Functionality can also be checked at the interface.
The self-product that guesses two markings of one byte stream, latches their first disagreement and accepts when both
markings are accepted with the latch set is empty exactly when the object is functional. The other half of being a
policy is that every nonempty accepted marking closes with a zero bit, which is read off the transitions into accepting
states. Nothing below inspects a policy's internals.

A marked string is best read as a scan's transcript. Take the token set $\{\code{ab}, \code{b}\}$ under maximal munch
and the input \code{bab}. Its unique segmentation is $\code{b} \cdot \code{ab}$, and its marked string pairs the bytes
\code{b}, \code{a}, \code{b} with the boundary bits $1, 0, 0$, one bit per byte saying whether a token boundary follows
that byte, so a set bit closes the current token and the next one begins at the following byte. The bit after the final
byte is always zero, since the string's end closes the last token and a boundary there would open an empty segment. The
verifier for this policy accepts exactly the marked strings whose bits close segments that are tokens, each the longest
token that begins at its position in the input, and it accepts nothing at all for an input like \code{ba}, which no
segmentation covers. The splitting companion's armed-run verifier carries the same information on the other side of each
byte, a mark on every token start, position zero included, so its mark at position $i+1$ is the bit here at position
$i$. A one-byte buffer translates either encoding into the other, so finite verifiability does not depend on the choice.
Everything the paper does is a question about the regular language of such transcripts.

Two conventions about the empty input are fixed here because later constructions lean on them. First, the empty input
carries no bits and admits one candidate segmentation, the empty sequence, but whether that segmentation is in a
policy's domain varies. The three tokenizer constructions below consume the empty input trivially and accept it. A
policy accepting nothing does not, and neither does a restriction to a schema that excludes the empty input. The
presentation-level reductions of Section~\ref{sec:real} need policies of exactly that kind. The verifier answers that
question from its own start state, and it must, since Lemma~\ref{lem:outside} reads the policy's domain off the
projection of the accepted language, and that projection contains the empty string exactly when the start state accepts.
Second, the start state is kept distinct from every state reachable after a byte. A fresh start state that no transition
re-enters ensures this. The point is that whether nothing has been consumed is a fact about the run, which a
configuration alone does not record. A derivative automaton for $\code{a}(\code{aa})^{*}$ returns to its initial
configuration after \code{aa}, and reading the return as an unconsumed state accepts two markings of that string. Under
the normalization the start state's acceptance is exactly the empty input's membership, and every product and
enumeration below uses it.

Three policies whose scanners share no code carry the interface, as a test of its generality. For literal maximal munch,
a verifier run pairs a candidate token with the length consumed so far. At the next proposed mark the segment the run
spans must be accepted by some token, and from then on the run is armed. A run that completes a token while armed has
found a longer match past a commit, which greedy scanning forbids. For regular-token maximal munch, the runs are
Brzozowski derivatives of the token expressions, armed the same way, and the tested sets include $\{a^{+}b,\ a\}$, whose
boundary bits no finite-state bounded-delay scanner can compute. For two-mode WordPiece, the verifier carries the mode
bit, and a closed segment flips it to continuation, which is the split of subword tokenizers between word-initial and
mid-word vocabularies. Each verifier is validated against an independent reference scanner, namely the direct greedy
scan, a longest-language-match reference over regular expressions, and the two-mode scan respectively, on every marking
of every string of bounded length. In each case the verifier accepts exactly the reference segmentation of each
tokenizable string and nothing for the untokenizable ones, over 1{,}530, 381, and 765 strings.

\section{The four decisions}\label{sec:four}


Four questions are decidable, three for every supplied verifier and equivalence for every pair of them, with the
counterexamples and constructive outputs stated below. The theorem decides them from supplied verifier presentations,
one finite construction per question. We first define the covering condition used in certification and in every transfer
law of Section~\ref{sec:transfer}.

\begin{definition}[Correct covering; certification]\label{def:cover} For a window $W$ of length $h$ and an origin
$0 \le o < h$, a policy $P$ \emph{covers} an occurrence of $W$ at position $p$ of a string in its domain
\emph{correctly} when $p+o$ is a token start of $P$'s segmentation and no token start lies in $(p+o,\,p+h-1]$.
Equivalently, the token covering the occurrence's final byte begins exactly $o$ bytes into it. $P$ \emph{certifies}
$(W,o)$ when it covers every occurrence of $W$ in every string of its domain correctly.
\end{definition}

\begin{theorem}[The four decisions]\label{thm:four} Let a well-formed functional finite verifier presenting a
segmentation policy be given as its states, its transition and acceptance functions, and its finite alphabet. The
following are decidable, uniformly in the verifiers given, and each answer comes with the object named.
\begin{enumerate}
\item \textbf{Certification.} Given a nonempty window $W$ and an origin $0 \le o < |W|$, whether the policy certifies
$(W,o)$, and when it does not, a string of the domain with an occurrence it covers wrongly.
\item \textbf{Boundary gap.} Whether the length of a boundary-free stretch is bounded over the policy's segmentations,
with the exact supremum when it is and a pumpable lasso of the verifier when it is not.
\item \textbf{Realizability.} Whether the domain is nonempty, and when it is, the mask that admits at each
\emph{completable} state, one from which an accepted completion exists, exactly the steps after which an accepted
completion remains, under which generation from any completable state never dead-ends.
\item \textbf{Equivalence.} Given two such verifiers, whether the two policies have the same domain and the same
segmentation on it, and separately each half of that: whether they segment every string of their common domain alike,
and whether their domains are equal. Each failure comes with a differing string, one in one domain and not the other or
one in both domains segmented apart.
\end{enumerate}
\end{theorem}

Three remarks fix the statement's terms. Throughout, \emph{boundary gap} names the stretch of the second item, between
consecutive token boundaries, and its supremum. The leading stretch, from position zero to the first boundary, and the
trailing stretch, from the last boundary to the end of the input, count the same way, end of input being an implicit
close, and a string with no marked boundary is a single stretch. A domain that is empty or holds the empty string alone
has no stretch, and its supremum is zero by convention, a value the decision returns although no stretch attains it.
Realizability guarantees that generation never blocks, and it says nothing about whether a run stops. An empty domain
answers no, and on a nonempty one the computed mask is the positive answer's output. The supplied verifier is needed
even for the boundary gap's yes-or-no, which Proposition~\ref{prop:gap-bounded} shows undecidable from a program
promised to compute a policy of the class.

Two small examples illustrate the decisions. Over $\{\code{ab}, \code{b}\}$, the window $(\code{b}, 0)$ is
\emph{refuted}, since the input \code{ab} is one token whose \code{b} sits mid-token, and the decision's counterexample
is that very string. The window $(\code{ab}, 0)$ is \emph{certified}, since \code{ab} occurs in tokens only
token-initially, so every occurrence in every tokenizable input begins a token. The two halves of the equivalence
decision separate cleanly on $\{\code{ab}, \code{b}\}$ against $\{\code{ab}, \code{b}, \code{a}\}$. Adding the one-byte
token creates a domain difference, the differing string \code{a}, tokenizable only on one side, while the boundary half
finds no \emph{boundary differential}, no commonly tokenizable string the two segment apart, because on every commonly
tokenizable string the added token never wins a longest match. The two policies segment their common domain alike and
differ only in which strings they tokenize. Domain difference and boundary differential are independent findings, and
the deciders separate them with differing strings.

\begin{proof} The verifier is finite, and each property is decided by a finite computation
over the verifier or a product with it, namely reachability or emptiness for certification and equivalence, a cycle
search and a longest path for the boundary gap, and a co-accessibility mask for realizability. Each returned string is
read back from its search path. Throughout, functionality is used freely. A policy assigns at most one segmentation per
string, and the verifier accepts exactly the graph, so an accepted marking \emph{is} the policy's segmentation of its
byte projection.

\emph{Certification.} Fix the window $W$ with $h = |W|$ and origin $o$. In a marked string the token covering position
$j$ begins at the last token start at most $j$, which is why the two readings of Definition~\ref{def:cover} agree.
Correct covering of the occurrence at $p$ comes down to two facts about the $h$ token-start indicators of the
occurrence's positions. Build the product of the verifier with a matcher that nondeterministically designates one
occurrence of $W$ in the byte projection and records those bits as it passes. A state is violating when its designated
occurrence has completed with the covering condition false. The pair is refuted exactly when a violating state can
complete to acceptance, which is reachability in a finite graph, and the reachability path projects to a violating input
with its own segmentation as the exhibit.

\emph{Boundary gap.} Restrict the verifier to its live states, those reachable from the start and from which acceptance
is reachable, since only live states occur in accepted runs. Consider the subgraph of transitions whose mark bit is
zero. If it contains a cycle through a live state, pumping the cycle inside any accepting run through that state
produces accepted markings with arbitrarily long boundary-free stretches, so the boundary gap is unbounded. The output
is the lasso, a pumpable family, no single string exhibiting unboundedness. If it is acyclic, each boundary-free stretch
is measured along its zero-edge path, its closing byte counted where one exists. An internal stretch of $k$ zero edges
ends in a mark-one byte and spans $k+1$ bytes; a final stretch ends at accepting end of input and spans its $k$ zero
edges alone, the terminal byte carrying the non-closing bit. The supremum is the larger of the two maxima over the
finite acyclic graph, computable and, wherever a stretch exists at all, attained. The empty and empty-string-only
domains take the conventional zero, which no stretch attains. Measuring the zero path alone would undercount every
internal stretch by its closing byte. The policy $\{\code{aab} \mapsto 010\}$ segments as $\code{aa}\,|\,\code{b}$ with
boundary gap two against a longest zero path of one.

\emph{Realizability.} Call a state completable, or co-accessible, when some marked suffix leads from it to an accepting
state. Completability is a reverse reachability fixpoint, computed once. The mask admits, at state $q$, exactly the
pairs $(\sigma, m)$ with $\code{step}(q, \sigma, m)$ completable. Suppose generation stands at a completable state $q$.
If $q$ accepts, generation may stop; if not, completability of $q$ names a marked suffix to acceptance, whose first step
the mask admits, so at least one admitted step exists and preserves completability. A masked generator therefore never
dead-ends, and the mask is computable from the fixpoint. The hypothesis that generation begins at a completable state is
needed, since a verifier accepting nothing is a policy of this class whose start state completes to nothing. Its mask
offers neither a stop nor a step, and the decision answers no there, since a guarantee would hold only vacuously. Once
co-accessibility is computed the mask follows directly, and realizability is kept among the four decisions because the
interface makes its guarantee available wherever the other three apply.

\emph{Equivalence.} Three terms are fixed first and used in these senses throughout. Boundary equivalence is agreement
of the two segmentations on the common domain. A domain difference is a string tokenizable under exactly one policy.
Full equivalence is boundary equivalence together with equal domains, and the theorem's fourth item asks about it first.
The reductions of Section~\ref{sec:real} refute it with differing strings that lie in only one domain.

Full equivalence is language equality, and that is the first decision. Functionality gives a policy exactly one accepted
marking per string of its domain, so a verifier's accepted language is the policy's graph, the marked strings it accepts
being the pairs of a domain string and that string's own segmentation. Two policies are therefore fully equivalent
exactly when their accepted marked-string languages are equal. In one direction, equal languages have equal byte
projections, and those projections are the two domains. On a string both domains contain, the marking the first side
accepts lies in the second side's language too, and the second accepts exactly one marking of that string, so the two
segmentations agree there. In the other, unequal languages carry a marked string one side accepts and the other does
not, whose bytes either leave one domain or are segmented apart, and either way the policies are not fully equivalent.
Both verifiers are deterministic over the alphabet of (byte, bit) pairs, so this is equality of two deterministic
automata. Build the product over that alphabet, a byte outside one policy's own alphabet sending that side to a dead
state without its verifier being stepped, and search the product for a reachable pair in which exactly one side accepts.
Breadth-first search returns the shortest such marked string, or reports that no such pair is reachable. The search
visits each product state once, so the decision costs the product's size. Its states are the pairs of the two verifiers'
states with a dead state adjoined on each side, and out of each pair it has one edge per bit for every byte of the union
alphabet.

The two halves are decided separately, each by a product of its own, and the same products read a differing string for
the full decision. The decision also reports whether a differing string separates the domains or the segmentations,
since the reductions of Section~\ref{sec:real} separate on the domain half and the pipeline applications need to know
which half fails. Boundary equivalence is decided by the first product. Run both verifiers over one byte stream with
independently guessed marks $(m_1, m_2)$ and a flag that latches when $m_1 \neq m_2$. A byte outside one policy's own
alphabet is outside its domain, so no jointly accepted stream contains one, and this product walks the intersection of
the two alphabets. If a run reaches joint acceptance with the flag set, functionality turns the two accepted markings
into the two policies' own segmentations of the same input, and they differ at the latched position, which gives a
boundary differential with the input synthesized from the path. Conversely a differential input yields such a run by
guessing both true markings. Emptiness of flagged joint acceptance is boundary equivalence.

Domain difference is the other half, and it needs a product of its own, since the first accepts only strings both
policies tokenize. That product walks the union of the two alphabets, the side whose alphabet lacks the byte recorded as
dead while the other runs on. Each policy's domain is the byte projection of its verifier's language, regular by
projection, so the difference of the two domains is regular and its emptiness decidable. The construction that yields a
differing string is the subset product. Track, for each side, the set of verifier states its surviving marking guesses
can occupy after the bytes read so far, which is that side's subset construction and finite. The string read so far is
in one side's domain exactly when some state in its set accepts, so a reachable pair of sets in which one side accepts
and the other does not marks a differing string, and the bytes labeling the path to it reconstruct the string.
Breadth-first search over pairs of subsets terminates and returns the shortest such string, or reports that the domains
coincide.
\end{proof}

The generic deciders are validated against the purpose-built ones they generalize. Certification is checked against the
pending-automaton product on 36 triples, and equivalence against the armed-run product on 9 regex pairs with every
differing string re-verified. The boundary gap is checked against the known literal boundary gaps, that is, each set's
longest token, and against the unbounded verdict on the unbounded-delay set. The realizability mask is checked to leave
an admitted step at every nonaccepting completable state, on one policy of each kind. The generic equivalence decision
also proves a literal policy equivalent across representations to the regular policy denoting the same token languages.

For maximal munch the splitting companion's two instruments, the product of two pending automata for literal policies
and the armed-run product for regular ones, are validated against brute force on 225 and 25 ordered pairs. The
language-equality decision is compared with the two-half decision on all 251 pairs, the two pools together with the
cross-representation pair. One product over marked strings returns a differing string exactly where the two-half
decision returns one, and each of the 228 shortest differing strings is replayed against the reference scans, 212 of
them separating the domains and 16 the boundaries.

The boundary gap is not the anchor-free span of the splitting companion, which counts the longest run of positions
carrying no anchor from a fixed inventory and can be unbounded where the boundary gap is bounded. Over the greedy token
set $\{\code{aa}, \code{a}\}$ the longest boundary-free stretch is two, the longest token, while every unary window
$\code{a}^{k}$ occurs in $\code{a}^{k+1}$ at starts of both parities, so by parity no window fixes the offset at which
its covering token starts, and the anchor-free span is unbounded. The execution refutes every one of the 10
window-and-origin pairs over that set's one-letter alphabet through length four.

\begin{theorem}[Cut independence]\label{thm:cut} Let a well-formed functional finite verifier presenting a policy $P$
be given as in Theorem~\ref{thm:four}, and let $(W, o)$ be a pair it certifies. It is decidable whether every certified
cut is independent: whether for every $x$ in the domain and every occurrence of $W$ in $x$ at $p$, with $c = p + o$,
$u = x[0..c)$ and $v = x[c..)$, the prefix $u$ lies in the domain with $P(u)$ equal to $P(x)$ restricted to $u$, and the
suffix $v$ lies in the domain with $P(v)$ equal to $P(x)$ restricted to $v$ and shifted by $c$. When some certified cut
is not independent, a shortest such $x$ is returned with the occurrence and the failing side or sides.
\end{theorem}

A certificate makes $c$ a boundary of the whole string's segmentation, and that is all it says. Whether each side scans
alone is a further question, and the decision separates the two. The construction is a product over the verifier in two
phases. Before the cut it carries the verifier's state and a matcher for $W[0..o)$, so that the cut, placed after a
marked byte, is where an occurrence may begin $o$ bytes earlier. At a nonzero cut the product holds two states for the
byte immediately before it. The whole-string run reads that byte with bit one and continues, while the prefix test reads
the same byte from the same state with bit zero, the last bit of the marked $u$ under the final-zero convention, and
asks whether the state reached accepts, which by functionality holds exactly when $P(u)$ is $P(x)$ restricted to $u$. A
failure is latched and the search goes on, because a policy whose cut state does not accept can still have $x$ in its
domain. A cut at $0$, from an occurrence at $0$ of a window certified at origin $0$, has no byte before it, so both
copies then begin at the start state, and the empty prefix fails exactly when the start state does not accept. After the
cut the same marked suffix is read by two copies of the verifier, one continuing and one restarted from the start state,
while a matcher confirms that $W[o..)$ follows at once. The restarted copy accepts exactly when $P(v)$ is $P(x)$
restricted to $v$. A path on which the occurrence has completed and the continuing copy accepts, with the flag set or
the restarted copy rejecting, projects to a counterexample with its own segmentation, and breadth-first search returns
the shortest. Write $Q$ for the verifier's state set and $\Sigma$ for the byte alphabet. With $h = |W|$, the product has
at most $|Q|(o+1) + 2|Q|(|Q|+1)(h-o+1)$ states, each with at most $3|\Sigma|$ out-edges, since besides its two mark
choices a byte may also offer the choice to place the cut.

Over 494 small universes the decision is held to brute force on 5{,}627 certified pairs, every verdict equal through
length 8 and all 62 counterexamples replayed by the reference scans. Every cut is independent on every tested certified
pair of a literal, regular or byte-pair policy, 5{,}256 pairs of 469 literal sets, 56 of the five regular sets and 48 of
the five byte-pair towers. These cases check the splitting companion's theorem for maximal munch, and a tower, with one
literal maximal-munch stage per merge, inherits it boundary by boundary. Independence fails where a scan carries state
across a boundary that the suffix alone cannot recover, as it does on 21 of the 59 certified WordPiece pairs, where the
restarted scan begins in the initial mode, and on 41 of the 68 certified pairs of the six threshold splices of
Section~\ref{sec:real}, where a side shorter than the threshold is segmented under the other policy.

Over $\{\code{a}, \code{ab}\}$ with continuation $\{\code{b}\}$ the pair $(\code{bb}, 1)$ is certified and \code{abb},
segmented $\code{ab} \cdot \code{b}$, is cut at $2$ into a prefix the policy segments the same way and a suffix \code{b}
that is not in its domain at all. For a policy that is not maximal munch, splitting at every cut of an inventory, a set
of certified pairs, needs more than the certificates, and the property it needs is exactly the one the next theorem
decides.

\begin{theorem}[Chunk independence]\label{thm:chunk} Let a well-formed functional finite verifier presenting a policy
$P$ be given as in Theorem~\ref{thm:four}, and let $I$ be a nonempty finite set of pairs it certifies, an
\emph{inventory}. It is decidable whether $I$ is chunk-independent: whether for every $x$ in the domain, with anchors
$c_1 < \dots < c_n$ the positions $p + o$ for every pair $(W, o)$ of $I$ and every occurrence of $W$ in $x$ at $p$, each
of the chunks $x[0..c_1)$, $x[c_1..c_2)$, \dots, $x[c_n..)$ lies in the domain with $P$ of it equal to $P(x)$ restricted
to it and shifted by its start, an anchor at $0$ leaving an empty first chunk, in the domain exactly when the verifier
accepts its start, and a string with no anchors, $n = 0$, its own sole chunk. When $I$ is not chunk-independent, a
shortest such $x$ is returned with its anchors and the chunks that fail. If every pair of $I$ is cut-independent in the
sense of Theorem~\ref{thm:cut} then $I$ is chunk-independent; the converse does not hold, even for an inventory of one
pair.
\end{theorem}

Parallel scanning cuts at every anchor of an inventory and scans each chunk alone, and since a middle chunk has a cut at
both ends, the exact condition it needs is chunk independence of the whole inventory, for which cut independence of
every pair in the sense of Theorem~\ref{thm:cut} is sufficient but not necessary. The implication is short. A string
with no anchors is its own sole chunk, which lies in the domain with its own segmentation. Otherwise the first and the
last chunk are the sides of the first and the last cut. For a middle chunk $x[c_i..c_{i+1})$, the prefix $x[0..c_{i+1})$
lies in the domain segmented as $x$ is, by the pair anchoring $c_{i+1}$. The occurrence anchoring $c_i$ lies within that
prefix, since its certificate makes the token covering the window's final byte begin at $c_i$ and that token ends at or
before the boundary $c_{i+1}$. The pair anchoring $c_i$, applied to the prefix, therefore puts the chunk in the domain
segmented as the prefix is.

The converse fails because the other anchors, another pair's or the window's own overlapping occurrences, can cut a
string so finely that no chunk is misread, although a whole side would be. Over the initial tokens
$\{\code{a}, \code{ab}, \code{b}\}$ with continuation $\{\code{a}, \code{b}\}$ the pair $(\code{aa}, 1)$ is certified
and not cut-independent, the suffix \code{ab} of \code{aab} reading as one token alone, while the inventory
$\{(\code{aa}, 1), (\code{aab}, 2)\}$ is chunk-independent, the second pair anchoring the \code{b} of every \code{aab}
so that every chunk a restarted scan could read differently is the single byte \code{a}. Over
$\{\code{a}, \code{aa}, \code{ba}\}$ with continuation $\{\code{a}\}$ the pair $(\code{aaa}, 2)$ is certified and not
cut-independent, the suffix \code{aa} of \code{aaaa} reading as one token alone, while the inventory of that one pair is
chunk-independent, its overlapping occurrences anchoring every position of a run of \code{a}'s from the third on, so
that every chunk after the first is a single byte.

The decision is a product over the verifier with two copies, one reading $x$ and one reading the current chunk. The
chunk copy is fed with a lag of $h$ bytes, $h$ the length of the longest window, so that when a byte reaches it every
occurrence that could anchor the position after that byte has completed. At an anchor the chunk copy reads the byte with
bit zero, its acceptance is tested, a refusal is latched, and the copy is restarted. Every state is tested as a possible
end of $x$ by committing the held bytes with no further input, and it is violating when the copy reading $x$ accepts and
some chunk has been refused. A product state records the two verifier states, the up to $\max(h-2,0)$ bytes kept ahead
of the lagged ones so that every occurrence can be matched, the up to $h$ marked bytes the lag holds, a bit for the
first commit and the refusal flag, so the product has at most
$4|Q|(|Q|+1)\sum_{k=0}^{\max(h-2,0)}|\Sigma|^k\sum_{j=0}^{h}(2|\Sigma|)^j$ states with $2|\Sigma|$ out-edges each, a
count every product of the sweep is checked against.

Over 498 small universes the decision is held to brute force on 29{,}532 inventories of one to three certified pairs,
and 29{,}530 verdicts agree with brute force through length 8, while the remaining two are refutations whose
counterexamples have length 9, past the brute-force bound. All 1{,}588 counterexamples replay on the reference scans,
and the 1{,}586 within the bound have the brute-force shortest length. The 27{,}896 inventories with every pair
cut-independent are all chunk-independent, and 48 inventories are chunk-independent although one of their pairs is not
cut-independent, 3 of them inventories of a single pair. Every tested inventory of a literal, regular,
threshold-restriction or byte-pair policy is chunk-independent, 23{,}807 inventories over the 469 literal sets, 784 over
the five regular sets, 1{,}541 over the six restrictions and 916 over the five byte-pair towers, which is the splitting
companion's theorem for maximal munch; 442 of the 485 WordPiece inventories and 822 of the 1{,}577 threshold-splice
inventories are not, for the reasons the single cuts fail.

An inventory is fit for parallel scanning exactly when this decision says so. For a policy that is not maximal munch
this decision is the question to ask, cut independence of its pairs being only the sufficient condition that
Theorem~\ref{thm:cut} checks one pair at a time.

\section{Where the class sits}\label{sec:real}


The finitely verifiable class strictly contains the online-scannable policies and is strictly contained in the policies
with a regular domain. Over the marking encoding it is exactly the rational segmentation policies. The undecidability of
general transducers concerns relations that are not functions, and what remains undecidable here sits in how a policy is
presented.

\paragraph{Online scanners are inside.} An \emph{online scanner} here is a deterministic finite-state device reading
bytes left to right and emitting each boundary bit within a bounded delay. Devices with unbounded buffering are excluded
from this definition. Every online scanner's policy is finitely verifiable. For a scanner with delay $d$, run it beside
the proposed marks, hold the last $d$ bits, and compare each emitted bit when its delay expires. The buffer is finite,
so the inner region of Figure~\ref{fig:landscape} sits inside the class by construction. Some policies have no online
scanner yet are finitely verifiable. Over the token set $\{a^{+}b,\ a\}$, the scan of $a^n$ commits nothing until the
input ends, and the scan of $a^n b$ commits everything as one token. The first commit's delay grows with $n$, so no
online scanner in the sense above computes its boundary bits. The armed-run verifier covers the set regardless,
validated on 1{,}020 strings, 255 for each of four regex sets, with every candidate marking tried. For this set,
repeated verifier states also let the instrument shorten certification counterexamples while preserving the violating
occurrence. The separation is not new, and neither is this kind of example. \citet{karttunen1996directed} observes the
same obstruction for the replacement of $a^{+}b$ by $x$, whose transducer is unambiguous and yet cannot be
sequentialized, because no finite-state device can accumulate an unbounded amount of delayed output. The four decisions
here hold on both sides of that frontier.

\paragraph{The frontier is decided.} Classical transducer results decide whether a supplied verifier's policy admits an
online scanner. A scanner has delay $k$ when it emits the bit of each position after reading at most $k$ further bytes,
the bits still pending at end of input being emitted there. Write $V$ for the supplied verifier. Replace every edge
$(\sigma, m)$ of $V$ by the transducer edge $\sigma \mid m$ and trim the result. This transducer $T_V$ is
letter-to-letter by construction and, since the verifier is functional, realizes the policy's partial boundary function,
which is length-preserving. By \citeauthor{choffrut1977sequential}'s theorem \citep{choffrut1977sequential}, in the form
\citet{beal2002determinization} state it, that function is subsequential exactly when $T_V$ has the twinning property,
which is decidable in polynomial time in the size of $T_V$ \citep{weber1995economy,beal2003squaring}. When it fails, the
failure is shown by a pair of same-input paths from the initial state reaching two states together with same-input loops
at them along which the relative output delay changes. When the test succeeds, take a trim subsequential realization.
Every cycle on one of its successful runs reads as many bytes as it writes bits, since pumping a cycle that did not
would produce an accepted pair of unequal lengths, so the lag between input and output is bounded.
\citet{frougny1998synchronisation} turn a subsequential transducer of bounded delay into an on-line automaton, one that
is letter-to-letter after an initial period, which is a scanner emitting each bit within a bounded delay. Conversely, a
bounded-delay scanner with its buffer of pending bits is a subsequential transducer of bounded delay. Whether a supplied
verifier's policy lies in the inner region of Figure~\ref{fig:landscape} is thus decidable in polynomial time from the
verifier.

Inside that region the delay is decided level by level. Projecting $T_V$ onto its bytes, each edge keeping its bit,
gives a nondeterministic automaton whose accepting runs are unambiguous by functionality. After trimming, two edges with
the same byte, source and target cannot carry different bits, since a common path into the source and out of the target
would otherwise give one byte string two accepted markings, so the projected automaton's choice of successor determines
the bit. In the other direction, a scanner with delay $k$ fixes the bit from the state and the lookahead alone. Let $q$
be a state of the trimmed automaton, $u$ a byte string reaching it, and $v$ and $v'$ continuations from $q$ to
acceptance that agree on their first $k + 1$ bytes. The words $uv$ and $uv'$ lie in the domain, each with one marking,
and they agree through the byte at which a delay-$k$ scanner emits the bit of the byte after $u$, so the scanner emits
the same bit for both, and that bit is the one each marking carries there. The bit at $q$ is therefore a function of $q$
and the next $k + 1$ bytes, and by determinism of $V$ so is the successor. When fewer than $k + 1$ bytes remain, the
lookahead is the whole remaining suffix, and acceptance is tested on that suffix alone and on no longer continuation it
begins. If the input ending there is accepted, its one marking fixes the bit and the successor, and if it is not, no
accepted run passes this way and the choice is free, as it is for any lookahead that begins no accepted continuation. A
missing transition goes to a rejecting sink, which makes the automaton total as the definition reads it. A scanner with
delay $k$ is exactly a $k$-lookahead delegator for that automaton, choosing the successor forced on every accepted
continuation that shares the lookahead, which \citet[Corollary~8]{loding2013delegators} decide in polynomial time for
each fixed $k$ over the fixed byte alphabet. A later erratum finds Lemma~3 of that paper faulty, confirms Lemma~7 and
Corollary~8, and proves Corollary~8 without that lemma \citep{loding2023errata}.

The levels are distinct, since for $k \ge 1$, maximal munch over $\{a,\ a^{k}b\}$ has delay exactly $k$. A token
starting at a byte is $a^{k}b$ exactly when the next $k$ bytes are $a^{k-1}b$, so $k$ further bytes settle every bit,
and fewer do not, since $a^{k+1}$ and $a^{k}b$ share their first $k$ bytes and differ in their first boundary bit. These
results locate finite verifiability relative to bounded-delay scanning. The decisions require only a verifier and apply
whether or not such a scanner exists.

\paragraph{Non-regular domains are outside.} The class also has an outer boundary that is simple to state.

\begin{lemma}[Non-regular domains are outside]\label{lem:outside}
If a policy's set of tokenizable strings is not a regular language, the policy is not finitely verifiable.
\end{lemma}

\begin{proof}
The unmarked projection of a regular marked-string language is regular, and the projection of the verifier's accepted
language is exactly the policy's domain.
\end{proof}

By the lemma, any policy whose domain is not regular sits outside the hypothesis, a policy whose domain is the
balanced-bracket language among them. The lemma draws that line and no more. In particular it says nothing about how a
policy is presented. Membership is determined by the marking relation, so a stack-defined device is neither enough to
exclude a policy, if the relation it computes is regular after all, nor enough to include one. A regular domain,
however, does not suffice for finite verifiability, as the next proposition shows.

\begin{proposition}[A regular domain does not suffice]\label{prop:regular-domain}
There is a segmentation policy whose domain is the regular language $\{\code{a}, \code{b}\}^{*}$ and which is not
finitely verifiable.
\end{proposition}

\begin{proof}
Take the tokens $\{\code{a}, \code{b}, \code{ab}\}$ and the policy that reads a leading \code{ab} as a single token
exactly when the rest of the input has the form $\code{a}^{n}\code{b}^{n}$, and otherwise cuts every byte into a token
of its own. Every string is tokenizable, since \code{a} and \code{b} are tokens, so the domain is all of
$\{\code{a}, \code{b}\}^{*}$. The bit following the first byte is zero on the strings
$\code{ab}\,\code{a}^{n}\code{b}^{n}$ and, by the final-zero convention, on the two one-byte strings, and on nothing
else. If the policy's marked strings formed a regular language, intersecting it with the regular set of markings whose
first bit is zero and projecting the bits away would make the set of those strings regular. But a regular language stays
regular when two strings are removed from it, and what remains is $\code{ab}\,\code{a}^{n}\code{b}^{n}$, which is not
regular.
\end{proof}

A regular domain is thus necessary but not sufficient for finite verifiability. The domain is the marking relation's
projection, so a non-regular domain forces a non-regular relation and puts a policy outside by Lemma~\ref{lem:outside},
while a regular domain, as the example shows, does not bring a policy inside. A supplied verifier does, and it is the
hypothesis the theorems use.

\paragraph{The class is a normal form.} It is tempting to place this class below a frontier of rational transductions,
where \citet{filiot2016equivalence} record inclusion, equivalence, and sequential uniformization as undecidable in
general, and to say that restricting to finite verifiability restores decidability. The encoding itself rules that
reading out. A boundary marking pairs each byte with exactly one bit, so the relation carrying an input to its marking,
a relation over $A^{*} \times B^{*}$ for the byte alphabet $A$ and the bit alphabet $B$, is length-preserving. By
Eilenberg and Schützenberger's theorem \citep{eilenberg1974automata} a length-preserving rational relation over
$A^{*} \times B^{*}$ is a rational subset of $(A \times B)^{*}$, equivalently realizable by a letter-to-letter
transducer. A rational subset of $(A \times B)^{*}$ is a regular language of marked strings, which is a finite verifier.
The converse is immediate, since a finite verifier over marked strings is already a letter-to-letter transducer, hence
rational. Over this encoding the two classes coincide.

\begin{proposition}[Collapse]\label{prop:collapse} For policies presented by boundary markings, one bit per byte, a
segmentation policy's marking relation is rational if and only if the policy is finitely verifiable.
\end{proposition}

\begin{proof}
A finite verifier is a deterministic automaton over the alphabet of (byte, bit) pairs, so its language is a rational
subset of $(A \times B)^{*}$ and the relation it defines is rational. For the converse, the marking relation is
length-preserving by construction, so by \citet{eilenberg1974automata} it is a rational subset of $(A \times B)^{*}$.
The same step can be checked in \citet{elgot1965relations}. By their Theorem~6.1 (p.~58), a rational relation being a
transduction in their terms, every bounded transduction, a length-preserving one among them, is finite-automaton
definable, and by their Proposition~4.4 (p.~55) a length-preserving finite-automaton definable relation is defined by an
elementary automaton, every edge of which carries one letter per tape (their Definition~4.1, p.~55). The rational subset
is a regular language of marked strings, and determinizing an automaton for it yields a verifier.
\end{proof}

The finite-verifier discipline is not a restriction chosen to make the questions answerable but a normal form for
rational segmentation policies, and the four decisions apply to every one of them. The undecidability the general
transducer results describe belongs to relations that are not functions, and a policy's marking relation is a function
by definition, which is why its equivalence was decidable by \citet{blattner1977single} before any encoding. Length
preservation adds the normal form, one verifier interface on which all four decisions and the calculus run.

\paragraph{Undecidability sits in the presentation.} The distinction between a supplied verifier and an arbitrary
program motivates the interface of Section~\ref{sec:interface}, and the undecidability reaches all four decisions. Every
decision here is uniform in its \emph{given} verifiers, whose four methods are the input and from which the products are
built. No decision is uniform in a policy presented as an arbitrary program promised to compute a policy of the class.
Effective finite presentations, regular expressions, automata, or letter-to-letter transducers among them, convert
uniformly into verifiers and are not the obstruction.

Fix a machine $M$ and an input $w$, and build threshold policies that simulate $M$ on the fixed $w$ for $|x|$ steps, so
that which branch governs $x$ depends on $|x|$ alone. For each fixed $M$ the result is one branch everywhere or a
regular length-threshold splice of two, a \emph{threshold splice}, hence finitely verifiable either way. Three splices
then carry the four decisions. Splice singleton splitting over $\{\code{a}, \code{b}\}$ into greedy scanning over
$\{\code{ab}, \code{a}, \code{b}\}$ after the halt. The splice flips certification of $(\code{ab}, 1)$, true on the
first branch and refuted on the second. It also moves the exact supremum the boundary gap decision returns from one to
two, both branches bounded, so what the splice reaches is the supremum the decision returns and not boundedness alone.
Deciding either from the program would decide halting. The simulated machine is normalized to halt only after at least
one step, so the first branch always governs some input.

Splicing refusal into singleton splitting, so that only inputs past the halting step are tokenized, makes the domain
empty exactly when $M$ never halts, so realizability decides halting. The third splice compares singleton splitting with
its splice into refusal after the halt, which is its restriction to inputs shorter than the halting step. The two are
fully equivalent exactly when $M$ never halts, so equivalence decides halting too.

The simulation runs on the fixed $w$ because, if $M$ were simulated on $x$, the set of inputs $x$ on which it halts
within $|x|$ steps need not be regular, and the constructed policy then need not be finitely verifiable. The shipped
instruments check the claimed differences between each pair of branches, and the reductions themselves are proofs.
Existential membership in the class is not a usable hypothesis; a supplied verifier is, and every theorem here is stated
of the verifier presentations it names. Figure~\ref{fig:landscape} draws the resulting landscape.

The first splice keeps both branches bounded, so it shows that no procedure computes the supremum from a program and
leaves open whether boundedness alone could be decided that way. Another branch after the halt closes that question, and
the certificate $(\code{ab}, 1)$ still flips.

\begin{proposition}[Boundedness from a program]\label{prop:gap-bounded} For a machine $M$ and an input $w$, let $P_{M,w}$ be the policy over $\{\code{a}, \code{b}\}$ that splits a string $x$ into single bytes when $M$ does
not halt on $w$ within $|x|$ steps, and otherwise takes a nonempty $x$ as one token, as regular maximal munch over the
token language $\{\code{a}, \code{b}\}^{+}$ does. Each $P_{M,w}$ is finitely verifiable. If $M$ never halts on $w$, then
$P_{M,w}$ certifies $(\code{ab}, 1)$ and its boundary gap is one; if $M$ halts on $w$, then $P_{M,w}$ refutes
$(\code{ab}, 1)$ and its boundary gap is unbounded. Hence no procedure that receives a program promised to compute a
finitely verifiable policy decides whether the policy's boundary gap is bounded.
\end{proposition}

\begin{proof}
Both branches tokenize every string over $\{\code{a}, \code{b}\}$, so the domain is all of them. If $M$ never halts on
$w$, every string is split into single bytes, the accepted markings are the empty string and the markings $1^{n-1}0$ for
$n \ge 1$, a regular language of marked strings, and every boundary-free stretch is one byte, so the boundary gap is
one. Every position is then a token start, so an occurrence of \code{ab} at $p$ has a token start at $p+1$ and none
after it inside the occurrence, and $(\code{ab}, 1)$ is certified. If $M$ halts on $w$ after $t \ge 1$ steps, the
strings shorter than $t$ are split into single bytes and every other string is one token. The accepted markings are the
empty string, the markings $1^{n-1}0$ of the strings of length $n$ with $1 \le n < t$, and the all-zero markings of the
strings of length at least $t$. That language of marked strings is regular because it is the union of a finite set and
the all-zero markings of length at least $t$. Each string has one marking in it, and every nonempty marking in it ends
in a zero bit, so a deterministic automaton for it, built from $t$ alone, is a verifier presenting $P_{M,w}$. Every
string of length $n \ge \max(t, 2)$ is one boundary-free stretch of $n$ bytes, so the boundary gap is unbounded, and
such a string beginning with \code{ab} has its position $1$ inside the token that begins at $0$, so $(\code{ab}, 1)$ is
refuted. A program for $P_{M,w}$, which simulates $M$ on $w$ for $|x|$ steps before choosing the branch, is computed
from $M$ and $w$, so a procedure deciding boundedness of the boundary gap from such programs would decide whether $M$
halts on $w$.
\end{proof}

\begin{figure}[ht]
\centering
\begin{tikzpicture}[every node/.style={font=\small}]
\draw[rounded corners=10pt, thick] (0,0) rectangle (14.0,5.6);
\node[anchor=north west] at (0.25,5.45) {segmentation policies};
\draw[rounded corners=9pt, thick] (0.5,0.4) rectangle (10.2,4.8);
\node[anchor=north west, text width=9.3cm] at (0.7,4.65)
    {finitely verifiable, equivalently rational by Proposition~\ref{prop:collapse}: the four decisions hold, closed
    under the calculus};
\draw[rounded corners=8pt, thick] (0.9,0.75) rectangle (6.4,3.05);
\node[anchor=north west, text width=5.2cm] at (1.1,2.9)
    {online scannable: bounded delay, decidable from the verifier};
\node[anchor=south west, text width=3.4cm] at (6.7,0.85)
    {$\{a^{+}b,\ a\}$: verifiable, not scannable};
\node[anchor=south west, text width=3.2cm] at (10.5,0.85)
    {balanced brackets: outside by Lemma~\ref{lem:outside}};
\node[anchor=south west] at (1.1,0.95) {literal maximal munch};
\end{tikzpicture}
\caption{The landscape the hypothesis draws. Verification is strictly weaker than scanning, and the frontier between
them is decidable. Every policy with a non-regular domain lies outside the class, including one whose domain is the
language of balanced brackets. The theorems hold throughout the class, its online-scannable policies included, and the
class is exactly the rational segmentation policies.}
\label{fig:landscape}
\end{figure}
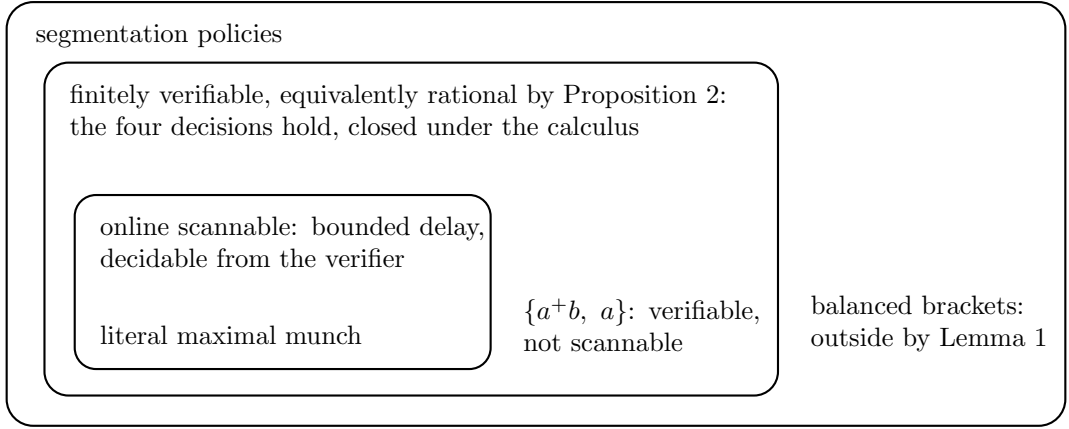

\section{The calculus}\label{sec:calculus}


Real systems rarely run one policy. A pre-tokenizer splits at class changes, a subword stage re-segments each run, a
schema restricts the inputs a stage can see, and a byte fallback totalizes a partial vocabulary. The class is closed
under all of it.

\begin{definition}[Labeler]\label{def:labeler} A \emph{labeler} for a policy is a deterministic finite automaton over
bytes, run from its initial state on each segment's bytes, that names the segment when it closes by a map from its
states to a finite name alphabet. A literal vocabulary supplies a labeler directly, by spelling, and a regular one by
least accepting index. A finite token alphabet alone does not, so the labeler is data that a stage's presentation must
supply.
\end{definition}

\begin{theorem}[Closure]\label{thm:calculus} Finitely verifiable policies are closed under four operations. A
\emph{pipeline} over a stage one presented with a labeler feeds each closed stage-one token to a second policy as its
name. A \emph{restriction} limits the inputs to a regular schema. A \emph{splitter pre-tokenizer} forces boundaries at
character-class changes and runs the inner policy on each run in isolation. A \emph{fallback union} totalizes a literal
vocabulary with single-byte tokens. Each composite verifier is constructed explicitly, and every composite inherits the
four decisions of Theorem~\ref{thm:four} unchanged, together with the cut and chunk independence decisions of
Theorems~\ref{thm:cut} and~\ref{thm:chunk} for the pairs it certifies and the inventories of them.
\end{theorem}

\begin{proof}[Proof sketch]
Of the four operations, pipelines need the main construction. Pipeline composition requires a supplied finite-state
labeler with a finite name alphabet, and stage two reads that alphabet. A literal vocabulary can use its token spellings
as names, so a spelling labeler's alphabet is stage one's finite token alphabet. A regular vocabulary may have
infinitely many spellings, such as the lexemes $a^n b$ of a regular stage, and no finite composite state can store an
arbitrary one, so its presentation must supply an effective finite-state labeler with finite range. The construction
takes the labeler as a parameter, so the closure holds for every labeler a stage is presented with.

Under the hypothesis a composite marked string carries only the composite boundaries, so the composite verifier must
judge stage one's segmentation without seeing its marks. It runs the subset construction on guessed marks, tracking the
finite set of triples of stage-one verifier state, labeler state, and stage-two verifier state consistent with the bytes
read and the composite marks seen, beside one bit recording whether any byte has been read at all. Whenever a tracked
guess closes a stage-one segment, the labeler's name for it is fed to stage two and the labeler resets.

At end of input a tracked triple that has read something feeds its final segment's label to stage two with the
non-closing bit before asking both stages' acceptance; one that has read nothing asks both stages about the empty input
and feeds no label. Acceptance requires a tracked triple in which both stages accept.

The consumed-input bit records whether a byte has been read, without inspecting the inner policy's start representation.
The tracked triples lie in a finite product of state sets, so their possible subsets form a finite set and the composite
verifier is finite. Each composite verifier accepts exactly the composite policy's markings, so it is functional and
well formed, and Theorems~\ref{thm:four} to~\ref{thm:chunk}, which take any such verifier, apply to it unchanged.

The other three constructions are short. Restrictions are the product with the schema's automaton, marks invisible to
the schema side. Splitters force a boundary at every character-class change and run the inner policy on each run in
isolation, and their composite verifier carries a one-byte lookahead buffer, enough to see the class change that forces
a boundary before committing the bit. At a forced boundary the composite bit is one, and the buffered byte reaches the
inner policy with bit zero, since its run ends there, so the inner policy must accept before it restarts on the next
run. Fallback unions are set union of a literal vocabulary with the byte alphabet, a plain policy again, total by
construction.
\end{proof}

Each construction is validated against a composite reference scan on every marking of every small string. The regular
pipeline, its labeler naming each closed segment by its least accepting language index, the same priority the reference
scan uses, is exact on 5{,}080 strings, the literal pipelines on 1{,}020 strings across four pipeline pairs, and the
splitters on 3{,}094 across three splitter configurations. Given the labeler that spells a segment back, the composite
is decided boundary-equivalent to the literal pipeline on every staged pair, and each of the two matches the same
two-stage reference scan over $508$ strings.

With a labeler that records only a segment's length class, stage two reads a coarser name alphabet instead of the inner
token set used by the literal pipeline construction. That difference is one of alphabets only. The composite built over
the length-class labeler is decided equal to maximal munch over
$\{\code{aa}, \code{ab}, \code{ba}, \code{bb}, \code{a}, \code{b}\}$, the single-stage policy it flattens to, and a
literal pipeline reaches that same policy by another route, so in this example the coarser labeler shrinks the stage-two
alphabet without producing a policy unavailable to literal pipelines.

Definition~\ref{def:labeler} says the labeler is part of stage one's presentation and not a consequence of it. The
following proposition makes that precise.

\Needspace{6\baselineskip}
\begin{proposition}[A boundary verifier does not determine the names]\label{prop:labeler}
A finite verifier for a policy's boundaries, together with the fact that its token names are finitely many, does not
determine which name a closed segment carries. Two ordered families can present one policy, so that a single verifier
accepts the markings of both, and still induce labelers that disagree on every segment. A second stage then tells the
two pipelines apart.
\end{proposition}

\begin{proof}
Take the ordered family, a finite list of token languages read in priority order, whose first token language is
$\{\code{a}\}$ under the name $X$ and whose second is $\{\code{a}\}$ under the name $Y$, over the alphabet
$\{\code{a}\}$. Nothing in the definition of an ordered family forbids repeating a language under two names. A literal
vocabulary names tokens by their spellings and cannot tell the two copies of $\{\code{a}\}$ apart, so the family is a
named regular one. Maximal munch with list priority closes a segment at every byte and names it $X$. Reversing the order
closes a segment at every byte and names it $Y$. The two presentations have the same domain and the same boundary
function, so they are one policy in the sense of Definition~\ref{def:fv} and one verifier accepts the markings of both,
while the names differ at every segment. A stage two that accepts $X^{+}$ and refuses any $Y$ then has domain
$\code{a}^{+}$ under the first labeler and the empty domain under the second, so the composite differs. The boundary
graph records where segments close but not which name each receives. It therefore omits the priority rule used to choose
each segment's name when token languages match equally.
\end{proof}

The calculus has two immediate uses. First, the literal pipeline of the preceding example is proved to agree with the
single-stage policy it flattens to, by the generic boundary decider of Theorem~\ref{thm:four} on the common domain,
which for these total policies is equivalence. Removing the token \code{bb} from that single-stage policy makes it
differ from the pipeline, and the decider returns \code{bb} as a differing string. Second, certification becomes
\emph{schema-relative}. The window $(\code{b}, 0)$ is refuted globally over $\{\code{ab}, \code{b}\}$ and certified on
the \code{b}-only schema, by the same decider on both sides of the restriction, so a certificate's meaning under a data
contract is computed by the same decider.

\section{Stage-local certificates}\label{sec:transfer}


Certificates computed for one stage survive composition under four laws. Each is proved below and checked by execution
on finite cases, against the deciders or, for the boundary-subset law, against an independent reference scan, and each
check is named where its counts are quoted. Covering an occurrence correctly, and certifying a window, are the terms of
Definition~\ref{def:cover}, under which the position $p+o$ the origin names is a token start and no token start follows
it inside the occurrence.

\begin{proposition}[Monotonicity]\label{prop:monotone}
If $P$ certifies $(W,o)$, so does its restriction to any regular schema $K$.
\end{proposition}

\begin{proof}
The restriction's domain is $\mathrm{dom}(P) \cap K$ and its segmentation of each string is $P$'s, so every occurrence
it must cover is one $P$ already covers correctly.
\end{proof}

Restriction only adds certificates. The global certified set is a subset of the restricted one on all 12 staged cases, 7
strictly, so a global certificate is reusable under every schema.

\begin{proposition}[Splitter forcing]\label{prop:splitter}
Let $S$ be the splitter over a class function $\kappa$ and any inner policy, and let $W$ have length $h \ge 2$ with
$\kappa(W_{h-2}) \neq \kappa(W_{h-1})$. Then $S$ certifies $(W,\,h-1)$.
\end{proposition}

\begin{proof}
The splitter cuts its input into maximal same-class runs and segments each run in isolation, so every token lies inside
one run and every position at which the class changes is a token start. At an occurrence of $W$ at $p$ the bytes at
$p+h-2$ and $p+h-1$ differ in class, so $p+h-1$ is a token start, and since it is the occurrence's final position no
token start lies beyond it inside the occurrence.
\end{proof}

A window whose final byte follows a class change is thus certified at that final position for every inner policy, as
checked on all 504 staged final-position cases. That is the splitting companion's anchor supply turned into a closure
operation. The analogous statement for an arbitrary origin is false, and the decider refutes it.

\begin{proposition}[Boundary subset]\label{prop:subset}
For a pipeline $P = P_2 \circ_\lambda P_1$ over a labeler $\lambda$, $\mathrm{dom}(P) \subseteq \mathrm{dom}(P_1)$, and
for every $x \in \mathrm{dom}(P)$ the token starts of $P(x)$ are a subset of those of $P_1(x)$.
\end{proposition}

\begin{proof}
By the construction of Theorem~\ref{thm:calculus}, $P$ is defined where $P_1$ is and $P_2$ accepts the label stream
$\lambda^{*}(P_1(x))$, and the pipeline segments $x$ by grouping consecutive stage-one segments as $P_2$ groups their
labels, so every composite token starts where some stage-one segment starts.
\end{proof}

A pipeline only removes boundaries, as checked on all 363 composite-tokenizable strings of the staged pairs, so at every
occurrence of a window the composite certifies, the position its origin names is a token start of the inner stage. The
fourth law is a decision rather than an inequality. Write $\Gamma$ for the labeler's finite name alphabet. A
\emph{flagged stream} is a word over $\Gamma \times \{0,1\}$ with exactly one flag set. When a flagged stream is the
label stream of $P_1(x)$ flagged on segment $k$, it selects $x$'s $k$-th stage-one segment. We write $(\gamma, k)$ for
the word $\gamma$ over $\Gamma$ flagged at index $k$, and say that it lies in $\mathrm{dom}(P_2)$ when $\gamma$ does.

\begin{theorem}[The lift]\label{thm:lift}
Let $P = P_2 \circ_\lambda P_1$ and fix a window $(W,o)$. Let $R \subseteq (\Gamma \times \{0,1\})^{*}$ be the set of
flagged streams obtained by taking $x \in \mathrm{dom}(P_1)$ and an occurrence at $p$ of $W$ in $x$ that $P_1$ covers
correctly, and flagging $\lambda^{*}(P_1(x))$ on the segment that begins at $p+o$. Let
$L \subseteq (\Gamma \times \{0,1\})^{*}$ be the set of flagged streams in $\mathrm{dom}(P_2)$ whose flagged index is
not a token start of $P_2$'s segmentation. Both are regular, $R$ effectively from $P_1$'s verifier, the labeler and the
window, and $L$ from $P_2$'s verifier alone. A stream in $R \cap L$ is the flagged stream of some string on which $P$
refutes $(W,o)$. If $P_1$ covers every occurrence of $W$ in every string of $\mathrm{dom}(P)$ correctly, in particular
if $P_1$ certifies $(W,o)$, then $P$ certifies $(W,o)$ if and only if $R \cap L = \emptyset$.
\end{theorem}

\begin{proof} \emph{$R$ is regular.} Run $P_1$'s verifier on bytes with guessed marks, the labeler beside it, reset at
every closed segment, and beside both the matcher of the certification proof, which tracks the window's match progress
and the token-start indicators of the last $h = |W|$ positions. Designate one occurrence nondeterministically at its
final byte, provided the indicators show a start at the origin's position and none after. The segment open at that byte
is then the one beginning at $p+o$, and it is flagged when it closes, at that byte's own mark, a later one, or end of
input. Emit each closed segment's label with its flag. The device is finite, so its accepted runs form a regular
language, and $R$ is the image of that language under the erasing morphism that keeps the emissions, regular again.

\emph{$L$ is regular.} Run $P_2$'s verifier on the stream with guessed marks, carrying whether the next index is a token
start, true at index zero and thereafter the previous mark. On the flagged symbol, latch whether that bit was false.
Accept when $P_2$ accepts with the latch set, and project the marks away.

\emph{Refutation.} Let $(\gamma, k) \in R \cap L$, and take $x \in \mathrm{dom}(P_1)$ and an occurrence at $p$ whose cut
$p+o$ begins segment $k$ with $(\gamma, k)$ as the stream flagged there, which exist since $(\gamma, k) \in R$. The
labeler need not be injective, so $x$ is one such string and not necessarily the only one. Since
$\gamma \in \mathrm{dom}(P_2)$, $x \in \mathrm{dom}(P)$; since $k$ is no token start of $P_2(\gamma)$, $p+o$ is no token
start of $P(x)$, and $P$ covers that occurrence wrongly.

\emph{Certification.} Suppose $R \cap L = \emptyset$, and take $x \in \mathrm{dom}(P)$ with an occurrence at $p$. By
hypothesis $P_1$ covers it correctly, so $p+o$ begins some segment $k$ and the stream flagged there lies in $R$. That
stream lies in $\mathrm{dom}(P_2)$ and not in $L$, so $k$ is a token start of $P_2$'s segmentation and $p+o$ one of
$P(x)$. No token start of $P(x)$ lies in $(p+o,\,p+h-1]$, every one being a start of $P_1(x)$ by
Proposition~\ref{prop:subset} and none of those lying there. Hence $P$ covers the occurrence correctly.
\end{proof}

Emptiness of $R \cap L$ is reachability in the product of the two automata. The product runs over the $R$-device's own
byte moves, the $L$ side stepping only at the labels those moves emit, so a reaching path projects to a refuting byte
string and not only to its label stream, and that projection is the counterexample. $R$ depends only on the labeled
inner stage and $(W,o)$, and $L$ only on the outer policy, so the two automata are stage-local. Thus $R$ can be built
once per window and origin and reused across outer policies.

The instrument builds both languages for regular inner stages, under the labeler that names a segment by its least
accepting language index, the priority the reference scan uses, and for literal outer stages. On all 190 (inner stage,
outer policy, certified window) cases the product's verdict agrees with the composite decider, generic certification on
the composite verifier of Theorem~\ref{thm:calculus}, and each of its 70 refutations reads back to a byte string that
the two-stage reference scan tokenizes and covers wrongly.

The product also finds refutations that a bounded enumeration misses. Over inner \code{a} and \code{b}, which the
least-index labeler names \code{0} and \code{1}, and outer tokens \code{0}, \code{1} and \code{0000000\,1}, the window
$(\code{ab}, 1)$ is refuted by the eight-byte counterexample \code{aaaaaaab}, the composite decider's own, which no
enumeration of streams to length seven contains.

\Needspace{6\baselineskip}
\begin{corollary}[The per-label condition]\label{cor:labels}
In the setting of Theorem~\ref{thm:lift}, with $P_1$ covering every occurrence of $W$ in $\mathrm{dom}(P)$ correctly and
$P_2$ literal over $\Gamma$, call a label \emph{realized left of the boundary} when some stream of $R$ carries it
immediately before its flagged segment. If no such label occurs as a non-final symbol of a token of $P_2$, then $P$
certifies $(W,o)$. When $P_1$ is literal, every label realized left of the boundary is the image under $\lambda$ of a
\emph{compatible left token}, an inner token $t$ whose suffix of length $\min(|t|, o)$ equals the suffix of $W[0..o)$ of
that length, so testing those images instead is a stronger sufficient test.
\end{corollary}

\begin{proof}
Suppose $(\gamma, k) \in R \cap L$. Index zero is always a token start, so $k \ge 1$ and $\gamma_{k-1}$ is realized left
of the boundary. Since $k$ is no token start of $P_2(\gamma)$, the token of $P_2$ covering index $k-1$ covers $k$ as
well, so $\gamma_{k-1}$ is a non-final symbol of it, against the hypothesis. Hence $R \cap L = \emptyset$ and
Theorem~\ref{thm:lift} applies. For literal $P_1$ the segment ending at $p+o$ is a token $t$ occupying
$[p+o-|t|,\,p+o)$, which meets the window's first $o$ bytes on $[\max(p,\,p+o-|t|),\,p+o)$, an interval of length
$\min(|t|, o)$, and agrees with them there, so $t$ is compatible and $\gamma_{k-1} = \lambda(t)$.
\end{proof}

Both literal-policy assumptions can be weakened to a known finite vocabulary. The first part needs only that every
segment of $P_2$ is a word of a known finite set $V$ over $\Gamma$, the hypothesis then reading that no label realized
left of the boundary occurs as a non-final symbol of a word of $V$. The second part needs only that every segment of
$P_1$ is a word of a known finite set $U$, every realized label then being $\lambda(t)$ for a compatible $t \in U$. The
proof above uses nothing more, and in that form the corollary covers towers, whose stages have finite vocabularies.

The following checked example shows why the condition concerns labels rather than spellings. Take inner
$\{\code{a}, \code{b}\}$ with every segment given one label, call it $Z$, and outer tokens $ZZ$ and $Z$. The inner
window $(\code{ab}, 1)$ is certified, and no compatible left token's spelling is carried non-finally, so a test that
read spellings instead of labels would grant the lift. The exact decision refutes it with \code{ab} as the
counterexample, since the outer stream is $ZZ$ and maximal munch takes it whole.

Compatible left tokens can also begin before the window. Over inner $\{\code{aa}, \code{a}, \code{b}\}$ the window
$(\code{ab}, 1)$ occurs at position one of $\code{aab}$, and the token ending at the position its origin names is
$\code{aa}$, although the window shows only that token's final $\code{a}$.

Nor is the condition necessary. Over inner $\{\code{a}, \code{b}\}$ with outer tokens $(\code{a},\code{a})$,
$(\code{a})$, $(\code{b})$, the window $(\code{ab}, 1)$ lifts although \code{a} is carried non-finally, because the
window fixes the following symbol and the merge the condition guards against is unreachable in its context. In general
the condition falls short because it reads what can stand left of the boundary and not what the window and the inner
stage allow to follow it, so a token carrying the left label non-finally can be unrealizable at the boundary. Evaluated
over the labels the product realizes, the condition grants no lift the composite refuses, and it refuses the lift on 41
of the 190 staged cases although the composite certifies them.

The lift also decides whether a certificate survives a dictionary update. Take the byte-pair table with the one merge
$\code{b}+\code{c}$ over the alphabet $\{\code{a}, \code{b}, \code{c}\}$ and extend it by the merge $\code{a}+\code{b}$
at the later rank, so that the extended table tokenizes \code{abc} as \code{a} then \code{bc}, the earlier rank winning.
The window $(\code{b}, 0)$ is certified under the one-merge table and refuted under the extended one by the
counterexample \code{ab}, one token, whose \code{b} is covered by a token beginning before it, while the window
$(\code{abc}, 1)$ is certified under both. The update is a pipeline in which the one-merge tower, whose single stage is
the literal set $\{\code{a}, \code{b}, \code{c}, \code{bc}\}$, sits below a literal stage over those token names that
carries the new merge as the token $(\code{a}, \code{b})$, each segment named by its spelling. Its reference
segmentation equals the extended table's on all 3{,}280 strings through length 7.

On the surviving certificate the per-label condition refuses the lift, since the label \code{a} is realized left of the
boundary and carried non-finally by $(\code{a}, \code{b})$, while the exact lift proves it, the product of the two
stage-local automata being empty and the composite decider agreeing. At an occurrence of \code{abc} the inner stage
supplies \code{a} followed by \code{bc}, so the outer token $(\code{a}, \code{b})$ cannot cross the marked boundary. The
product built for $(\code{b}, 0)$ refutes that window by the counterexample \code{ab}.

The two greedy disciplines separate on the same window. Literal maximal munch over the extended table's final vocabulary
$\{\code{a}, \code{b}, \code{c}, \code{ab}, \code{bc}\}$ tokenizes \code{abc} as \code{ab} then \code{c}, so the generic
decision certifies $(\code{abc}, 1)$ on the tower and refutes it on the literal policy with \code{abc} as the
counterexample. A certificate is a statement about one policy, and here the byte-pair table and literal maximal munch
over its final vocabulary are two different policies.

\begin{proposition}[Stack transitivity]\label{prop:stack}
Let $C = T \circ G$ be a tower of pipelines, $T$ reading $G$'s segments through a labeler as in Theorem~\ref{thm:lift},
whose innermost stage certifies $(W,o)$. Then $C$ certifies $(W,o)$ if and only if $G$ covers every occurrence of $W$ in
every string of $\mathrm{dom}(C)$ correctly and $R_G \cap L_T = \emptyset$, with $R_G$ and $L_T$ the languages of
Theorem~\ref{thm:lift} for $G$ and $T$.
\end{proposition}

\begin{proof}
By Proposition~\ref{prop:subset} at each level, the token starts of $C(x)$ lie among those of $G(x)$, which lie among
the innermost stage's, and the innermost certificate leaves no start in $(p+o,\,p+h-1]$ for any occurrence at $p$ in any
$x \in \mathrm{dom}(C)$, so $G$ covers such an occurrence correctly exactly when $p+o$ starts a token of $G(x)$, and $C$
exactly when $p+o$ starts one of $C(x)$. If $G$ covers every occurrence in $\mathrm{dom}(C)$ correctly,
Theorem~\ref{thm:lift} with $P_1 = G$ and $P_2 = T$ gives the equivalence. Otherwise some $p+o$ starts no token of
$G(x)$, hence none of $C(x)$, so $C$ refutes the window and both sides fail.
\end{proof}

The proposition thus poses a separate question at each level of a tower, with the quantifier running over the whole
tower's domain, and both premises are necessary. A top stage that rejects the stream carrying a violation of $G$ hides
it from $C$, so the quantification must run over $\mathrm{dom}(C)$ and not $\mathrm{dom}(G)$. A top stage that merges
the window's bytes into one token repairs a violation of $G$, so without the innermost certificate the equivalence
fails. Both examples are among the staged cases, and the equivalence is checked by enumeration on 60 three-stage cases.

\section{The weighted theory}\label{sec:weighted}


The decisions say whether; deployment asks how often. The hypothesis of this section covers every well-formed functional
verifier presenting a policy, in the sense of Theorem~\ref{thm:four}, under a rational, unifilar, almost surely stopping
finite-state source, unifilar meaning that the current state and the byte emitted determine the next state. A memoryless
byte source with rational byte probabilities and a positive rational stopping probability is the one-state case, and
Proposition~\ref{prop:weighted-general} shows that the theorem holds without unifilarity.

Token counts below follow the final-zero convention of Section~\ref{sec:interface}. For a nonempty string, the token
count is the number of set boundary bits plus one, and the empty string has zero tokens. Differential mass between two
policies is the mass of the strings of their common domain that they segment apart. A per-byte rate divides an expected
count, such as the expected number of mismatched boundary bits, by the expected length, which keeps it rational where an
expectation of per-string ratios need not be.

\begin{theorem}[Exact weighted shadow]\label{thm:weighted} Let $V$ be a well-formed functional finite verifier
presenting a policy, and let $S$ be a finite-state byte source with rational transition and stopping probabilities,
unifilar, and almost surely stopping. The following are rational: the tokenizability mass of $V$ under $S$; the
unnormalized expected length and expected token and boundary counts contributed by accepted strings; and, for a second
such verifier $V'$, the differential mass and the unnormalized expected length and number of mismatched boundary bits on
the common domain. Each is zero when the product has no live state. Otherwise each is the start-state coordinate of the
unique solution of an affine system $x = Ax + b$ over the live states of the product of $V$, or of $V$ and $V'$, with
$S$, where the spectral radius of $A$ is strictly below one. Every rate that is a quotient of two of these expectations
is rational wherever its denominator, an expected length, is positive.
\end{theorem}

\begin{proof}
Each mass and each unnormalized expectation is a vector indexed by the live states of the product of the verifier with
the source, the states reachable from the start from which joint acceptance is reachable, where the verifier accepts and
the source can stop. For the pair quantities the product runs both verifiers over one byte stream with their marks
guessed in pairs and their first disagreement latched. Each such vector satisfies an affine fixpoint $x = Ax + b$ with
coefficients drawn from the byte probabilities and stopping mass, and each rate is the ratio of two start-state
coordinates of such solutions. When no state is live the system is empty, and the mass, the expected length and the
expected counts are zero. Those systems are absorbing chains only in the deterministic presentation, since a raw row can
exceed one.

The bound on the spectrum comes from functionality. On a raw marked verifier both marks of one byte can lead to live
states, as they do over $\{\code{a}, \code{ab}\}$ after reading $\code{a}$, and rows of $A$ can then exceed one, so the
bound on $A^{n}$ has to come from multiplicity, since row sums give none. Two markings of one byte prefix that reach the
same live state extend, by liveness, along one completion to two accepted markings of one string, which functionality
forbids. Hence the alive markings of a prefix occupy distinct states, at most $N$ of them for $N$ live states, and the
mass of the paths of length $n$ is at most $N$ times the probability that the source emits at least $n$ bytes. For the
quantities solved over the product of two verifiers a marking is a pair of markings, and the same argument runs on
whichever component the two differ in.

From every state it enters with positive probability, a finite-state source that stops almost surely has a positive
probability of stopping within $m$ steps, $m$ its state count. Those states are finitely many, so the least of these
probabilities is a $\delta > 0$, and the probability of emitting at least $n$ bytes decays geometrically over blocks of
$m$ steps, by a factor of $1 - \delta$ per block, so its sum over $n$ is finite.

The Neumann series $\sum_n A^{n}$ therefore converges from the start state, hence from every live state, each reachable
from the start with positive mass, and the spectral radius is strictly below one, so the solution is unique. Where a
deterministic online presentation exists and every source state stops with positive probability, one successor per byte
keeps every row at most one minus the stopping probability and the same bound follows in a line. Rational coefficients
keep the unique solution rational, and a quotient of two rationals with positive denominator is rational. \end{proof}

\begin{proposition}[Positivity]\label{prop:positivity} Under the hypotheses of Theorem~\ref{thm:weighted}, the
differential mass of $V$ and $V'$ under $S$ is positive exactly when the boundary half of the equivalence decision of
Theorem~\ref{thm:four}, run on the restrictions of $V$ and $V'$ to the support of $S$, returns a differing string.
\end{proposition}

\begin{proof}
The support of $S$, the strings it emits with positive probability, is a regular language, recognized by the source's
states with its positive-probability transitions, accepting where the stopping probability is positive. Differential
mass is a sum of nonnegative terms, one per string of the common domain that the two policies segment apart, each that
string's probability under $S$, so it is positive exactly when some such string lies in the support. A differing string
through a zero-probability byte carries no mass. The restrictions of $V$ and $V'$ to the support are verifiers by
Theorem~\ref{thm:calculus}, with the common domain cut down to the support and the segmentations unchanged, so the
boundary half run on them returns a differing string exactly when such a string exists.
\end{proof}

The measure and the logic thus agree about zero, on the boundary half alone, since the measured event is disagreement on
the common domain. Two policies whose domains differ while their shared segmentations agree carry zero differential mass
and are still separated by the domain half, which the law does not speak about. Under a full-support source, one in
which every byte and the stop carry positive probability, the restriction is the identity and the law takes its simplest
form.

The theorem itself holds for sources in which one state and one byte lead to several next states, provided each
transition emits exactly one byte.

\begin{proposition}[Sources need not be unifilar]\label{prop:weighted-general} Theorem~\ref{thm:weighted} and
Proposition~\ref{prop:positivity} hold for every finite-state byte source $S$ with rational probabilities on its
byte-labeled transitions and on stopping that emits exactly one byte per transition and stops almost surely, whether or
not $S$ is unifilar.
\end{proposition}

\begin{proof}
The proofs of Theorem~\ref{thm:weighted} and Proposition~\ref{prop:positivity} use unifilarity, the fact that a byte
prefix determines the source's run, in two places, the product's definition and the bound on the spectrum, and both are
rebuilt here to allow several source runs for one prefix. Build the product over pairs of a source state and a verifier
state, with one edge for every source transition of positive probability and every mark the verifier admits on its byte,
weighted by the transition's probability, and trim it to its live states. A string's probability is the sum of the
weights of the source runs that emit it and stop, and functionality gives its bytes at most one accepted marking, so the
masses and expectations are again start-state coordinates of affine systems $x = Ax + b$ over the live states.

Fix a run of $S$ through $n$ transitions. Two distinct marking paths along it that reach one verifier state at a live
product state extend, along one source continuation and one marked completion to joint acceptance, to two accepted
markings of one string, which functionality forbids. The live marking paths along the run therefore occupy distinct
verifier states, at most $|Q|$ of them for the verifier's state set $Q$, and summed over the runs the mass of the
product's paths of length $n$ is at most $|Q|$ times the probability that $S$ emits at least $n$ bytes. For the pair
quantities a path carries a pair of markings, and two distinct pairs reaching one pair of verifier states at a live
product state differ in a component that would receive two accepted markings of one string, so at most $|Q|\,|Q'|$ pairs
live along each run, $Q'$ the second verifier's state set. The block bound in the proof of Theorem~\ref{thm:weighted}
uses only that $S$ is a finite Markov chain on its states that stops almost surely, so the probability of emitting at
least $n$ bytes decays geometrically over blocks, the Neumann series converges from the start state and from every live
state, the spectral radius of $A$ is below one, and each solution is unique and rational.

For positivity, the source's positive-probability transitions, accepting where the stopping probability is positive, now
form a nondeterministic automaton, whose language, the support of $S$, is still regular, so the proof of
Proposition~\ref{prop:positivity} runs unchanged on the restrictions to it.
\end{proof}

The executable interface of this paper remains unifilar. Its source objects present one successor per state and byte,
and the solver and its checks below run on those, so Proposition~\ref{prop:weighted-general} is proved here and not
claimed as implemented.

Each system is solved exactly over the rationals, every coordinate of the residual $x - (Ax + b)$ is checked to be zero,
and on 30 staged triples the dynamic program agrees exactly with per-string enumeration. The deterministic engine serves
the literal instances, while the regular-token example of Section~\ref{sec:real}, which has no online scanner, and the
two-mode instance are decided by the trimmed product alone, cut down to its live states.

The raw product is validated separately for each family. The memoryless literal cases and the literal pairs are compared
with their exact deterministic engines, and the two-state literal cases are checked by exact inverse and residual
identities and by finite multiplicity and path-mass tests. For the regular and two-mode single-policy cases, and for the
policy pairs outside the literal family, finite enumeration with exact tail bounds brackets the answers.

Unifilar Markov sources carry every quantity over by the same methods. The block-survival bound above replaces the
memoryless one-step tail, the solver refuses a source that does not terminate and gives its reason, and on a one-state
source it returns exactly the rationals of the memoryless solver it generalizes.

The positivity law is checked in its plain form on 100 ordered pairs under a full-support source, and in its
support-refined form, with the right-hand side computed by the paper's own restriction calculus, on 60 combinations of a
policy pair and a source, including two in which the policies differ globally but have differential mass exactly zero,
because the source's support excludes every differing string.

Two closed forms show what the solver computes. Both are outputs of the same solved systems, reported without separate
enclosures of their own. The one-merge extension $\{\code{a}, \code{b}, \code{ab}\}$ to
$\{\code{a}, \code{b}, \code{ab}, \code{ba}\}$ under a memoryless source has differential probability exactly $4/11$ and
per-byte mismatch rate $1/5$ under $P(a) = P(b) = 2/5$ and $P(\mathrm{stop}) = 1/5$. The analogous extension
$\{s, q, o, qo\}$ to $\{s, q, o, qo, os\}$ over the three classes structural, quote and other, under corpus digram
statistics, has differential mass about $0.063$, exactly $3396201545233/54186662216209$, against $833/2561$, about
$0.325$, under the matched uniform source. These two exact masses differ because the sources weight strings differently,
even though their stopping probabilities and expected lengths match. The digram conditionals are unsmoothed, both
sources stop with probability exactly $1/8$ at every state, the corpus source starts in the class of the corpus slice's
first byte, and the uniform comparator's expected length, exactly seven, equals the corpus source's. Corpus structure,
on this slice, concentrates mass on strings where the two vocabularies agree.

\section{Instances and applications}\label{sec:instances}

The series' main empirical results are reported in the two companion papers. The recovery companion measures certified
recovery against the classical conventions on archived campaigns, and the splitting companion measures anchor supply,
splits the recovery campaign's archived corpus at certified anchors with a threaded scanner, and decides the recovery
campaign's change of whitespace treatment with a synthesized differing string. This paper presents instances of the
abstraction and measures certificate transfer under dictionary growth.

A vocabulary upgrade from 320 to 352 tokens produces a boundary differential, a differing string of length two. The
tokens are byte-pair merges over the full byte alphabet, trained on the JSON sample the programs read, simdjson's
\code{twitter.json}. That file is fetched from its pinned release, held to its digest and read byte for byte. The added
merges change the segmentation of a two-byte input, showing how an upgrade can invalidate cached segmentations.

Byte-pair encoding is an exact instance of the class. Canonical BPE has the finite-state constructions of
\citet{cognetta2025tokenization} and \citet{berglund2024bpedfa}. For a proper dictionary, one whose every multi-symbol
merge part is an earlier merge's product, the construction of \citet{berglund2024bpedfa} is a subsequential transducer
from strings to tokenizations by Theorem~4 of that paper's extended version. Writing each emitted token $u$ as the mask
$0^{|u|-1}1$ and closing the output with a zero bit in place of the last one keeps that transducer subsequential and
length-preserving, so the bridge of Section~\ref{sec:real} places such a policy inside the online-scannable region, and
the four decisions apply to it directly. Canonical BPE is nonetheless not in general a maximal-munch policy, for the
reason Section~\ref{sec:related} records.

The lift of Theorem~\ref{thm:lift} decides whether each certificate of a growing dictionary survives its next merges.
Appending merges is one way pretrained tokenizers are adapted. \citet{purason2026teaching} continue a tokenizer's BPE
training on new data, \citet{smith2026inplace} continue its merges in place, and \citet{herold2025vocabulary} add domain
tokens under a guarantee that no input tokenizes longer, but none of them asks which certificates survive. The
exploration decides it on byte-pair dictionaries of 64 merges extended by the next $k$ merges of their own training, for
$k$ from 1 to 16. It runs two dictionaries side by side, one trained on the corpus and the other the first merges of the
GPT-2 table, applied to raw bytes without GPT-2's pre-tokenizer. In the GPT-2 table, 1 of the first 8 appended merges
and 3 of the first 16 have the space byte as their left part.

The verifier is a deterministic tower, one single-lookahead transducer per merge composed into one automaton, since the
subset-construction tower behind the byte-pair policies of Section~\ref{sec:four} passes 5{,}000 states within two
merges. Over the bytes that occur in some merge, with one inert stand-in for all other bytes, the deterministic tower
has 4{,}063 states for the corpus dictionary and 561 for the GPT-2 prefix. The reduction is sound because a byte that no
merge names can only be a single-byte token and no merge joins across it, and it is checked against the full byte base
on 48 sampled pairs. The tower requires that every merge produce a token distinct from the alphabet and from every
earlier product, as both trained tables do. The verifier is held to the reference encoder on exhaustive sweeps over the
five byte-pair universes of Section~\ref{sec:four} and 12 random tables, 2{,}555 and 22{,}746 strings, to the subset
tower on those universes, and to the reference markings of 171 held-out pieces, which the longest extension's tower
accepts while refusing every one of their 8{,}037 single-mark perturbations.

The inventory, frozen before any extension is decided, holds every window of width one to three, with every origin, in a
16 KiB slice held out from training, 3{,}740 pairs for the corpus dictionary and 2{,}684 for the GPT-2 prefix. Of these
the base towers certify 1{,}174 and 902, with the 2{,}566 and 1{,}782 refutations replayed by the encoder. The anchor
supply is measured on the next 16 KiB.

Each certificate is then decided three ways under every extension, by the composite decision on the extended tower, by
the exact lift and by the compatible-token test. For the lift, the inner stage is the base tower and the outer stage is
the tower of the appended merges over base-token names. Their pipeline has the extended tower's marked language, an
equality checked at $k = 16$, and $R$ is built once per certificate from the base tower and $L$ from the outer tower's
verifier. The compatible-token test is the stronger second part of Corollary~\ref{cor:labels}, in the finite-vocabulary
form stated after it. The outer vocabulary contains each base token, including each single byte, as a one-symbol word
and each appended product as the sequence of base tokens it concatenates.

The test is coarser than the per-label condition it bounds, which reads the labels realized left of the boundary. Over
$\{\code{a}, \code{b}\}$ with the base merge $\code{a}+\code{a}$ and the appended merges $\code{a}+\code{b}$ and
$\code{a}+\code{aa}$, the window $(\code{baaa}, 3)$ is certified by both towers and by the exact lift. The only label
realized left of the boundary is \code{aa}, and the per-label condition grants it, while the compatible tokens include
\code{a}, which the new token $(\code{a}, \code{b})$ carries non-finally, and the test refuses it.
Table~\ref{tab:transfer} has the counts. The exact lift equals the composite decision on every certificate of every
extension, and the test grants nothing the lift refuses. Of what the test discards, the lift keeps 474 certificates of
the corpus dictionary at $k = 16$ and 236 of the GPT-2 prefix.

On the corpus dictionary, the test's anchor supply decreases at every tested extension. At $k = 16$ the test supplies
313.8 anchors per KiB against the exact lift's 446.6, the longest anchor-free stretch is 42 bytes against 30, and the
lift recovers 2{,}126 of its 7{,}146 anchor positions on the evaluation slice, positions at which the test supplies no
anchor.

\begin{table}[tp]
\centering
\small
\begin{tabular}{llrrrrrrr}
\toprule
Dictionary & $k$ & \multicolumn{2}{c}{Kept} & Recovered & \multicolumn{2}{c}{Anchors per KiB} & Anchor gap &
Longest \\
\cmidrule(lr){3-4} \cmidrule(lr){6-7}
  & & exact & test & & exact & test & p99 & stretch \\
\midrule
JSON & base & 1{,}174 & & & 482.4 & & & \\
 & 1 & 1{,}168 & 1{,}132 & 36 & 480.1 & 476.4 & 17 / 17 & 22 / 22 \\
 & 2 & 1{,}162 & 1{,}068 & 94 & 478.0 & 462.8 & 17 / 17 & 22 / 24 \\
 & 4 & 1{,}152 & 969 & 183 & 472.2 & 428.2 & 17 / 17 & 22 / 24 \\
 & 8 & 1{,}136 & 799 & 337 & 461.9 & 365.9 & 20 / 21 & 30 / 32 \\
 & 16 & 1{,}111 & 637 & 474 & 446.6 & 313.8 & 20 / 23 & 30 / 42 \\
\addlinespace
GPT-2 & base & 902 & & & 927.9 & & & \\
 & 1 & 898 & 846 & 52 & 923.7 & 914.2 & 2 / 3 & 5 / 7 \\
 & 2 & 895 & 811 & 84 & 922.2 & 902.3 & 2 / 3 & 5 / 7 \\
 & 4 & 889 & 784 & 105 & 920.6 & 887.8 & 2 / 3 & 5 / 7 \\
 & 8 & 878 & 710 & 168 & 914.2 & 615.7 & 2 / 11 & 5 / 20 \\
 & 16 & 867 & 631 & 236 & 909.0 & 574.9 & 2 / 11 & 5 / 20 \\
\bottomrule
\end{tabular}
\caption{Base certificates of a dictionary of 64 merges that survive its next $k$ merges. The composite decision and the
exact lift agree on every certificate and share the columns headed exact, while test gives the compatible-token result.
Kept counts the survivors in the frozen inventory, and recovered counts the certificates the lift keeps but the test
discards. Anchors per KiB, the 99th-percentile gap between consecutive anchors, and the longest anchor-free stretch are
measured on the held-out evaluation slice, the last two in bytes. A cell with two figures gives exact followed by test,
and the base row has no extension to compare. The JSON dictionary is trained on the corpus, and the GPT-2 dictionary is
the first merges of the GPT-2 table, whose appended space-led merges begin at $k = 8$.}
\label{tab:transfer}
\end{table}

\Needspace{8\baselineskip} On the GPT-2 prefix, the test's anchor supply decreases slightly through $k = 4$, to 887.8
per KiB, then drops to 615.7. This drop occurs when the first space-led merge arrives at $k = 8$, while the exact lift
goes from 920.6 to 914.2. From there on the test refuses all 11 base certificates whose boundary follows a space, the
space being carried non-finally by the new token, and the lift, which asks instead whether the merge is ever realized
across the boundary, keeps every one of them. At $k = 16$ the exact lift supplies 909.0 anchors per KiB against the
test's 574.9, its p99 anchor gap is 2 bytes against 11, and its longest anchor-free stretch is 5 against 20, and the
lift recovers 5{,}346 of its 14{,}544 anchor positions.

Across all five extensions, the exact lift costs less even after including its one-time construction, 2{,}863.3 against
16{,}112 seconds of summed task time on the corpus dictionary and 652.6 against 1{,}502 on the GPT-2 prefix. For GPT-2
the first extension alone does not repay the build, the lift taking 425.1 seconds including the build against the
composite decision's 249. The stage-local automata of the base tower are built once, in 1{,}819 and 391 seconds of
summed task time, and excluding that construction a lift product over them is cheaper than deciding the extended
composite from scratch by a factor of 13.7 to 16.3 on the corpus dictionary and 4.7 to 7.3 on the GPT-2 prefix, in
summed task elapsed time. Over a sample of 24 certificates and the five extensions, reuse took 108.4s against 658.3s on
the corpus dictionary and 26.4s against 66.3s on the GPT-2 prefix. The compatible-token test is a sound first filter
that needs no product search, but after sixteen appended merges it discards 29.8 and 36.8 percent of the anchor
positions the exact lift keeps on the two dictionaries, against at most 9.3 percent through four.

WordPiece, the subword tokenization method of \citet{schuster2012japanese}, has a real-vocabulary instance too, at two
depths. The policy is the partial two-mode longest-match lexical policy over the BERT-Base, Uncased vocabulary of
\citet{devlin2019bert}, read on lowercased raw text, its plain entries the initial tokens and its \code{\#\#} entries
the continuation tokens. It is not the end-to-end tokenizer of BERT's reference implementation \citep{bertcode}, since
it omits the basic tokenizer's cleaning, its accent stripping by Unicode decomposition, its splitting on whitespace and
punctuation and its handling of CJK characters. It also omits the \code{[UNK]} fallback with its limit of 200 characters
on word length, and it fails where no token matches. A vocabulary does not determine its inference method, and
\citet{uzan2024greed} compare several over the same vocabularies, so the policy is named by its scan and not by its
vocabulary alone. \citet{song2021fast} compute that scan in linear time from the vocabulary's trie with failure links in
the manner of Aho and Corasick, while here a finite verifier recognizes its marked scans and the decisions quantify over
every string.

The vocabulary file opens with 999 reserved entries, its special tokens and unused slots, and 997 single characters, and
the first depth takes the next 4{,}096 entries in the file's token-ID order over those 997 characters. Reduced jointly
with the policies it is compared with below, its alphabet is a quotient of 41 symbols, namely 39 named characters, 38 of
them named by its own tokens, and two stand-ins. The remaining characters are grouped by whether each compared policy
holds them as initial and as continuation tokens, and each of the two groups that occur has one stand-in. At this depth
\code{<1>} stands for the 957 characters that every compared policy, the prefix's own among them, holds as single
initial tokens and none as continuation tokens, and \code{<2>} for the one character subscript two, which the deepest
comparison also holds as a continuation token.

The reduction is sound. A character that no multi-character token of any compared policy contains can be consumed only
as a single-character token, in the modes its membership vector allows, and no token begun before it reaches across it,
so replacing every such character by another with the same vector leaves every scan's success or failure and every
boundary unchanged. A refuting text over the full base maps to one over the quotient with the same boundaries and the
same occurrence, and a refuting quotient text lifts to a raw one by restoring the raw window's characters at the
distinguished occurrence and writing any representative of its class elsewhere, so a canonical window's verdict is the
verdict of every raw window it stands for. A check of 51 pairs decided both ways, the raw window over all 997 characters
and the canonical one over the quotient, confirms it.

Over the quotient the prefix has a verifier of 29{,}786 states, on which the four decisions take seconds to about a
minute. Of 4{,}747 window and origin pairs from a held-out slice of \code{twitter.json}, the prefix certifies 2{,}065
and refutes 2{,}682, every refutation replayed. Its boundary gap is 15, the length of \code{representatives}, with no
lasso, and the decisions return a realizability mask of 26{,}211 states that admits a next step at every nonaccepting
retained state. The vocabulary one entry deeper is separated from the prefix by the appended word \code{electoral},
which the prefix spells \code{el|e|c|tor|al}.

The two-mode policy and maximal munch over the same initial tokens segment \code{<1>00} differently, and only maximal
munch tokenizes \code{<1><1>}. The modes also show in the cut decision. Of the 476 occurring certificates, those whose
window occurs in some domain string, 466 are cuts that are not independent, each failing on the suffix side only, and a
shortest counterexample among these failures is \code{<1>ca}.

At full lexical depth the policy uses all 28{,}526 entries after the single characters, over the same 997-character base
and excluding the 999 reserved entries. Its verifier is built from its two tries, without exploring the scan's
configurations. Built for the prefix vocabulary, this trie-based construction and the tabulated verifier agree on
corresponding transitions. The full-depth verifier has 138{,}510 states over a quotient of 71 symbols, 70 named
characters and \code{<1>}, which at this depth stands for the other 927 characters, single tokens in both modes here and
initial ones only at the 4{,}096 prefix. Each quotient symbol is a single-character token in both modes, so every string
over the quotient alphabet is in the policy's domain.

On the same 4{,}747 pairs the full-depth policy certifies 1{,}230 and refutes 3{,}517, every refutation replayed. The
boundary gap is 18, the length of \code{telecommunications}, and the search for a lasso takes 708.1 seconds and finds
none. The mask keeps 136{,}766 of 138{,}510 states and admits a next step at every nonaccepting retained state. The
full-depth policy is separated from the 4{,}096 prefix by \code{02} on the boundaries and \code{<1><1>} on the domain,
and from maximal munch over its own initial tokens by \code{<1>0s} on the boundaries alone, the two domains coinciding.
Of the 590 occurring certificates, 196 are cuts that are independent, every one of them with a stand-in in its window,
and 394 are not, again each on the suffix side only.

Within maximal runs of explicitly named characters, which exclude the characters a stand-in represents, anchors are few.
The 16{,}384-character evaluation slice has 1{,}434 such runs, and counting only in the runs the policy tokenizes whole,
the prefix policy certifies 597 anchor positions and the full-depth policy 333, or 37.3 and 20.8 per 1{,}024 characters
of the full slice, none of them inside a run the scan reads as one token.

When runs may also contain the characters a stand-in represents, the raw slice has 1{,}049 maximal alphabet runs, within
which every quotient symbol is a single-character token in both modes at full depth, and the full policy certifies
4{,}853 positions across them, or 303.3 per 1{,}024 characters of the full slice, the other 5{,}003 characters of which
lie outside the alphabet.


The munch lexer library \citep{munch} implements the interface at its release 2.3.0 with a deterministic trim automaton
over marked symbols built from any transition table, the armed run of Section~\ref{sec:interface} over a byte automaton,
and the verifiers the two tokenizer instances above are decided on, the merge tower of canonical byte-pair encoding as
\code{merge\_tower()} and the trie chain of two-mode WordPiece as \code{trie\_chain()}. The library also implements the
four decisions of Section~\ref{sec:four} and the cut and chunk independence decisions of Theorems~\ref{thm:cut}
and~\ref{thm:chunk} as \code{dependence()} and \code{chunk\_dependence()}.

Its test suite holds each of them to reference answers this paper's exploration writes for it. The tests replay every
counterexample of a refuted window or a failing cut or chunk, every differing string of a failed equivalence and every
lasso, from both the library and the fixtures, through a reference scan, a reference encoder or the verifier. The
exception is the cut fixture's uncertified pairs, which the tests require to be refused without reading the refuting
texts recorded for them. Verdicts that return no string and boundary-gap suprema are compared, and a realizability mask
is checked against co-accessibility state by state in the fixture of the four decisions but only for its presence and
size in the merge-tower and trie-chain fixtures.

The fixture of the four decisions has 474 universes, every set of one to three distinct words of length one to three
over two symbols and the five regular sets of the edit sweep, and carries 16{,}116 certification decisions, 474
boundary-gap decisions of which 2 are lassos, 474 realizability verdicts and 1{,}818 equivalence pairs. The cut fixture
decides every window of length one to three at every origin over 489 universes, 5{,}517 certified cuts independent and
62 not. The chunk fixture decides 29{,}085 inventories over 494 universes, 1{,}759 of them not chunk-independent. Both
fixtures are the independence sweeps above with their five byte-pair universes removed, whose tabulated towers run to
thousands of states. The chunk fixture adds one universe of its own, every nonempty string over $\{a, b\}$ segmented
byte by byte, whose start state does not accept.

The merge-tower and trie-chain fixtures hold the two constructions to the reference encoder's markings of 10{,}177
strings over 14 byte-pair tables and the reference scanner's markings or refusals of 13{,}713 strings over 16 two-mode
vocabularies, and pose the four decisions on each.

Given a build of the library at that release, the gate's library stage regenerates the five fixtures, checks that each
matches the library's copy byte for byte, runs the library's verifier, merge-tower and trie-chain tests and its supply
probe, and writes the record this instance reports. Without such a build it reports the archived record and names the
regeneration it did not run.

On \code{twitter.json} under the JSON grammar the library's verifier examines 19{,}434 windows of one to three bytes at
53{,}074 window origin pairs and certifies 284 of them. It agrees with the library's online window route on every pair
except 17 the route refuses conservatively, and exact search confirms each of those certificates. The certified
positions are 43{,}043, or 69.79 per KiB, and the longest anchor-free stretch is 464 bytes.

Related questions about token boundaries arise in language-model serving. The finite-state tokenizer of
\citet{kaplan1998tokenizing} emits its output at a pinch point, a position where all tokenization paths reaching it can
be extended in exactly the same ways. That condition concerns one input's configurations, whereas a certificate here
asks whether one fixed window and origin determine a boundary at every occurrence across the whole domain.

The following systems resolve boundaries for individual inputs. IBM's speculative stream scanning partitions a stream at
likely token boundaries, scans the chunks in parallel and confirms or refutes each guessed boundary
\citep{vonpraun2011speculative}. A natural-language segmenter proposes segmentations and submits each to a linguistic
knowledge component for validation, proposing further ones until the input is validly segmented or no further
segmentation is possible \citep{pentheroudakis2006tokenizer}. A context-aware tokenizer combines a record that left and
right context automata build with its token automaton, segmenting left to right by longest match
\citep{walther2003context}.

An XML preparser runs a copy of its automaton from every possible state at a chunk's start and lets the preceding
chunk's actual end state select the run \citep{chiu2014metadfa}. An editor re-lexes the text an edit affects with a
batch lexer, comparing the new token stream with the old until the two agree \citep{vandevanter1998incremental}.
Incremental lexing restarts from recorded scanner states and reproduces the batch token stream exactly
\citep{wagner1997lexing}, and parallel lexing for parsing lexes chunks under alternative readings and reconciles them at
the joins \citep{barenghi2015parallel}. Parallel tokenization compares the tokenizations of overlapping chunks and
restarts with a doubled chunk when it finds no overlap token at the same position in both \citep{shao2026lopt}. Session
continuation matches fresh against cached token records and accepts a splice only through a long enough equal run
containing a family-specific stable boundary \citep{zhang2026toktier}. In generative steganography the sender
re-tokenizes the generated text after each token to obtain the receiver's view and resets only where that view diverges
from its prediction, suspending the check while it emits pre-segmentation anchor tokens until the next token outside
that family has fixed the boundary \citep{wang2026retoksync}.

For byte-pair encoding under a proper dictionary, the kind standard training produces, appending a suffix can change the
original prefix's tokenization only within a bounded distance of its end \citep{berglund2023formalizing}, the
tokenization is maintained incrementally \citep{jiang2026incremental}, and its streaming delay is bounded from a known
beginning \citep{mamouras2026bpedelay}. Three implementations precompute sufficient cut rules. Gigatoken uses absent
vocabulary byte pairs before selected punctuation, arguing that every merge result is a vocabulary piece
\citep{rod2026gigatoken}. SentencePiece Lite states the analogous character-bigram rule for normalized input
\citep{sentencepiecelite2026}. \citet{xu2026parallel} checks the merge table for cuts across which no merge can join.
GPU implementations of byte-pair encoding run its merges in parallel within thread blocks, dropping regex
pre-tokenization at a small loss in generation quality \citep{you2025blockbpe}, follow GPT-2's merge rules and produce
the same tokens as a CPU version on the sequences tested \citep{kadamba2026gputok}, or hold the token sequence as a
linked list and trade minor discrepancies in the tokenized output for throughput \citep{kovilur2026llbpe}.

The byte-pair bounds and the three static cut rules give guarantees across inputs for one tokenizer family, while the
other methods surveyed here determine the boundaries of the input they process. Certification asks, for a policy
presented by a supplied finite verifier and a proposed window and origin, whether every occurrence throughout the domain
pins that boundary. This predicate is decidable before any input is known. The static cut rules guarantee boundaries
across contexts for fixed patterns; the contribution here is an exact decision for arbitrary candidate windows and for
policies given by finite verifiers.

A prompt cache of a prefix's segmentation is sound for a set of continuations exactly when every continuation in the set
retokenizes the prefix identically. The next proposition gives this guarantee under maximal munch for the continuations
that keep the certified window in place and that the policy tokenizes completely.

\Needspace{6\baselineskip}
\begin{proposition}[Guarded cache soundness]\label{prop:cache} Let $P$ be maximal munch over a literal or regular token
set and let $(W,o)$ be certified by $P$. For every string $x$ of $\mathrm{dom}(P)$ with an occurrence of $W$ at $p$, the
prefix $x[0,\,p+o)$ is in $\mathrm{dom}(P)$ and $P$ segments it as $P(x)$ does, so a cache keyed on the cut $p+o$ is
sound under every continuation that keeps the window in place and that the policy tokenizes completely.
\end{proposition}

\begin{proof}
Certification makes $p + o$ a boundary of the complete scan of $x$. Every token the scan of $x$ commits before $p + o$
ends at or before it, and every candidate the scan of the prefix sees was already a candidate of the scan of $x$, so the
prefix's scan commits the same tokens and completes, which is the splitting companion's split theorem
\citep{nidhogg2026splitting} for one cut.
\end{proof}

The guarantee is stated for maximal munch, for continuations that keep the window, and for caches keyed on cuts, and the
examples below show why. A continuation that replaces the window's bytes is outside the guarantee. Maximal munch over
\code{a}, \code{ab}, \code{b}, \code{c} certifies \code{c} at origin zero, and the prefix \code{a} scans as one token
before \code{c} and as part of \code{ab} before \code{b}. Certification alone does not guarantee even the guarded form
for arbitrary policies, since a certified window pins the covering token's start and not the prefix's segmentation. A
policy with a two-string domain can certify a shared separator in both strings while segmenting the shared prefix
differently. Consistent restarts at every boundary do not guarantee the guarded form either. Over the alphabet \code{a},
\code{b}, \code{c}, take the policy that splits into single bytes unless the input ends in \code{c}, and scans by
maximal munch over \code{a}, \code{aa}, \code{b}, \code{c} when it does. It is finitely verifiable, restarts
consistently at every boundary, and certifies \code{b} at origin zero everywhere, yet the prefix \code{aa} before that
certified cut scans as two tokens in \code{aab} and as one in \code{aabc}.

The realizability decision governs a masked generator that must emit only completable token sequences, and its mask
provably never dead-ends. For a tokenizer upgrade, the equivalence decision returns a differing string exactly when the
upgrade changes the domain membership or the segmentation of some string, and a changed segmentation can invalidate a
cache keyed on cuts. By Proposition~\ref{prop:labeler} an upgrade can preserve every boundary and still rename every
segment, so a cache keyed on token identifiers needs the labeled comparison, while the boundary comparison here serves a
cache keyed on cuts. The decisions of Theorem~\ref{thm:four} provide these guarantees at that scope.

\section{Related work}\label{sec:related}

This section names the prior result behind each part of the paper and says what remains this paper's own.

\paragraph{The class is classical.} Definition~\ref{def:fv} names a known object. Writing $\mathrm{conv}(x,y)$ for the
convolution that pairs two strings position by position, a function $f$ is \emph{automatic} when
$\{\mathrm{conv}(x, f(x)) : x \in \mathrm{dom}(f)\}$ is regular \citep[Definitions~2.1 and~2.2]{case2013automatic}. The
relations so recognized are the synchronized rational relations of \citet{frougny1993synchronized}, recognized on
convolutions in which the shorter component is padded with blanks, equivalently realized by letter-to-letter automata
with a terminal function whose values are rational residuals
$(S \times \{\varepsilon\}) \cup (\{\varepsilon\} \times T)$, $S$ and $T$ rational, so that residual input remains on at
most one tape \citep[Definition~4.1]{frougny1993synchronized}.

A policy's boundary function maps a byte string to an equal-length bit string, so the convolution needs no padding and a
marked string \emph{is} that convolution. A policy is therefore finitely verifiable exactly when its boundary function
is automatic.

The standard consequences come with the class. Automata-recognizable relations are effectively closed under the Boolean
operations and both quantifiers \citep[Theorem~4.4]{khoussainov1995automatic}, so anything first-order definable from
finitely many automatic functions and relations is automatic again. Automatic structures also have decidable first-order
theories \citep[Corollary~4.2]{khoussainov1995automatic}, a result that goes back to \citet{hodgson1982direct}. This
paper claims none of that. The theory of automatic structures decides sentences, and an accepted string can be extracted
from any automaton it constructs. These general closure and decidability results do not by themselves state the content
of Sections~\ref{sec:four} through~\ref{sec:weighted}, including the exact supremum the boundary gap decision returns,
the realizability mask, and the transfer of certificates under the operations that build real tokenizers, which is this
paper's center.

\paragraph{What each decision inherits.} Each of the four decisions has a predecessor, three resting on earlier results
and certification refining the synchronizing pairs of code theory.

\emph{Equivalence.} Single-valued transducers have decidable equivalence, proved by \citet{blattner1977single} from
their Corollaries~2 and~3 (pp.~314--315). A verifier is a letter-to-letter transducer emitting one bit per byte,
single-valued because the policy is a function, so the decision already follows from that line of work.
\citet{veanes2012sft} state the shape the proof takes as well, that deciding $A \equiv B$ ``reduces to two independent
tasks'', domain equivalence and equality on the common domain. Their decision is non-constructive, and a separate
semi-decision procedure of theirs generates a shortest input on which the two outputs differ, so a differing string of
that kind is prior too. Their definition of single-valuedness admits partial functions and leaves the representation of
outputs free, so their results also cover outputs represented as boundary sets.

The route the proof above takes first is classical too. A marking language is regular and a verifier is deterministic
over marked symbols, so full equivalence is equality of two deterministic automata, which the product over the symmetric
difference decides by the decision problems of \citet[Theorem~10 and Corollary~10.1]{rabin1959finite}. Beyond that, the
two-task formulation tells whether a differing string is a domain difference or a boundary differential, and the
reductions of Section~\ref{sec:real} use domain differences.

\emph{Realizability.} The mask is co-accessibility, which supervisory control calls nonblocking
\citep[Section~6]{ramadge1987supervisory}. Constrained decoding compiles a target language against a language model's
vocabulary and masks the next token by the edges of the automaton that results \citep{koo2024automata}, and with the
automaton trimmed to its co-accessible states that mask admits exactly the tokens after which an accepted completion
remains. The proof above says so where it is used, and the decision here computes that mask from the supplied verifier
through the same interface as the other three.

\emph{Boundary gap.} Whether the boundary gap is bounded is a limitedness question. Read a boundary as a counter reset
and a byte as an increment, and a bounded boundary gap is an instance of limitedness for R-automata, whose counters are
incremented, reset or left unchanged. Limitedness asks whether some bound makes the language accepted within it equal to
the whole accepted language, and \citet[Theorem~4]{abdulla2008rautomata} decide that problem along with universality.
Beyond that verdict the decision here returns the exact supremum under the conventions of Section~\ref{sec:four}, and a
pumpable lasso when the gap is unbounded. A deterministic finite graph makes both elementary. Exact bounds for counter
automata over infinite words are computed as well as decided \citep{colange2015bounds}, and the value returned here is
that quantity for a deterministic marked graph, with this paper's counting conventions and a lasso. The nearest
tokenizer-specific quantity among the sources surveyed is the token-neighbor distance of \citet{li2026streaming}, which
measures how far past a token's end a maximal-munch scanner reads before that token's maximality is settled, from a
known start. The boundary gap needs no known start, measuring the longest boundary-free stretch anywhere in an accepted
segmentation, and the splitting companion separates the token-neighbor distance from certification and the anchor-free
span on one token set \citep{nidhogg2026splitting}.

\emph{Certification.} The certificate refines code synchronization. In the general convention of
\citet[p.~373]{berstel2010codes}, a pair $(x, y)$ of words of $A^{*}$ is synchronizing for $X^{*}$ when $uxyv \in X^{*}$
forces $ux \in X^{*}$ and $yv \in X^{*}$ for all $u$ and $v$ in $A^{*}$, so an occurrence of $xy$ in a word of $X^{*}$
fixes a factorization cut between its halves in every context. Over $X = \{ab, ba\}$ the pair $(b, b)$ is synchronizing
\citep[Example~10.1.2]{berstel2010codes}, though neither half lies in $X^{*}$.
\citet[Definition~2]{carpi2017synchronizing} restrict synchronizing pairs to $X^{*} \times X^{*}$. A code has finite
synchronization delay when some number of codewords on each side makes every such pair synchronizing
\citep[Definition~1]{restivo1975synchronization}. For a regular code, whether that delay is bounded is decidable through
the syntactic monoid of $X^{*}$ \citep[Th\'eor\`eme~2]{lassez1976messages}. For a regular synchronizing code the pair
test is finite-state too. Take a transitive, synchronizing, unambiguous automaton for $X^{*}$ whose base state is its
only initial and final state. Under those hypotheses a pair $(x, y)$ of $X^{*} \times X^{*}$ is synchronizing exactly
when two sets meet only in that base state, the states reachable by an $x$-labeled path from some state and the states
from which a $y$-labeled path exists \citep[Lemma~2(4)]{carpi2017synchronizing}.

The bounded-window form is synchronization with finite delay under a morphism, and \citet[Example~21]{fici2025morphisms}
contrast a language on which it fails with one on which it holds. Under the Thue--Morse morphism the doubled letters fix
the cut between them wherever they occur, yet on the images of $a^{n}b$ and $b^{n}a$ the windows free of such a pair
grow without bound, while every circular factor of at least five letters of the image of a word in $\{ab, ba\}^{*}$
holds one.

Prefix-synchronized codes enforce the property by starting every codeword with a marker that occurs at no other offset
of any concatenation of codewords \citep{morita1996prefix}, a window certified at origin zero by construction of the
code. The window paper of this series places the certificate in this lineage \citep{nidhogg2026splitwindows}, and the
definitions give the exact relation over a prefix code read as a token set. A window occurring in a word of $X^{*}$ is
certified exactly when every occurrence of it in every word of $X^{*}$ has a factorization boundary at the origin and
its part from the origin on is a prefix of a codeword. The first condition is the synchronizing-pair condition, in the
general convention, for the window's two parts $W[0..o)$ and $W[o..)$; the second, which certification adds, pins the
covering codeword, since a codeword beginning at the cut cannot end inside the window's tail without being a proper
prefix of the codeword that tail begins.

Neither condition is about the pieces a worker scans. Over the prefix code $\{\code{ab}, \code{cd}\}$ the window
$(\code{bc}, 1)$ is certified, yet its halves \code{b} and \code{c} are not words of $X^{*}$. Cut independence of
Theorem~\ref{thm:cut} concerns the whole input's prefix and suffix instead, $u\code{b}$ and $\code{c}v$ for an
occurrence in $u\code{bc}v$, each of which must scan alone with the segmentation it inherits, and chunk independence
asks the same of every chunk cut at an inventory's anchors.

Where data-parallel finite-state machines enumerate a chunk's possible entry states \citep{mytkowicz2014data} or
speculate on them \citep{zhao2014challenging}, a chunk-independent inventory makes the restart from the start state
correct, so no entry state is enumerated or guessed. Programmable speculation lexes a chunk from an entry state
predicted by lexing a few characters before it, validates the prediction against the predecessor's actual end state at
run time and re-executes the chunk on a miss \citep{prabhu2010speculative}. The boxing check of
\citet[Section~3.4]{lester2013boxing} validates one known join instead, the concatenation of two analyzed string
abstractions, requiring that every state the lexer may be in at the end of the left one emit a token and restart on
every character that may begin the right one, so that concatenation moves no token boundary.
\citet[Section~5.2.3]{lester2016flow} give the condition in full, requiring of every lexer state possible after the left
abstraction that either it already emits a token or every possible first character of the right abstraction immediately
makes it emit the same token, so that the join is a lexeme boundary. The speculation and the boxing check both validate
a seam at a narrower scope, a speculated state at run time or a known join of analyzed fragments, where a certificate
pins the boundary at every occurrence of its window throughout the domain.

A related notion outside code theory is the resolving block of \citet[Definition~3.3]{adler1983sliding}. For a factor
map of subshifts, all admissible preimages of such a block agree at a selected coordinate. A certificate instead
concerns finite accepted markings, in each of which the token covering the window's final byte starts at the origin.

In constrained decoding, for a fixed suffix under an unknown left context, \citet[Section~5.3.1]{melcer2024constrained}
compute every position of the suffix at which a lexical token begun before it may end, where a certificate asks that one
origin be forced at every occurrence. This paper decides that predicate for an arbitrary supplied verifier, over token
sets that need not be codes and over the domain the supplied verifier accepts, with the covering token's start fixed at
the proposed offset within the window in every occurrence. The decision is an emptiness question once the badly covered
occurrences are recognized, and its proof builds that recognizer for a supplied window and origin.

\paragraph{What the calculus inherits.} The closure results also have predecessors. Rational relations compose
\citep[Theorem~4.13]{elgot1965relations}, which is the pipeline clause once a supplied finite-state labeler makes the
intermediate alphabet finite. The labeler is the clause's hypothesis, since a boundary verifier and a finite set of
names do not by themselves recover which name a segment carries, which is Proposition~\ref{prop:labeler}. Restriction to
a regular schema is intersection with a lifted schema automaton \citep[Theorem~6]{rabin1959finite}. Byte fallback
totalizes a literal vocabulary the way \citet{cognetta2025tokenization} assume the atomic alphabet sits inside the
vocabulary, and the enlarged vocabulary is again finite and literal. Lemma~\ref{lem:outside} is the contrapositive of
the classical fact that a rational relation has a regular domain.

The lift computes an inner stage's realized label streams once and then intersects them with an outer stage's condition,
and \citet{park2025flexible} use the same split between precomputation and later checks at another interface. They
precompute from a maximal-munch lexer composed with a language model's vocabulary the terminal sequences each vocabulary
token can realize from each lexer state. For a context-free parser they precompute many sequences that every stack
accepts and many that no stack accepts, finding the rejected ones through a finite-state overapproximation that erases
the stack. Only the remaining sequences are checked against the stack while decoding. \citet{dong2025xgrammar} decide a
grammar's context-independent accepted and rejected tokens offline and check the rest against the live stack, the same
split between precomputation and checking while decoding, and \citet{li2026psc} compile the acceptance condition of
every vocabulary token into one classifier of the parser stack, so that one examination of the stack yields the whole
mask.

The lift's inner language is instead the set of flagged label streams of whole stage-one segmentations, and its outer
language the set of flagged streams that a finite-state outer stage accepts without starting a token at the flag. The
emptiness of their product decides whether a certificate transfers, where those systems decide which token a mask
admits.

The closure operations themselves are all known, and the section's contribution is to carry each one out on the
verifier, so that the composite is a verifier again and all six decisions run on it unchanged, and to show the
label-recovery hypothesis necessary.

\paragraph{What the weighted theory inherits.} The product of a verifier with a rational finite-state source is a
weighted automaton over the nonnegative rationals, so each mass of Section~\ref{sec:weighted} is the total weight of a
rational series, the series giving each string the weight of its accepted runs. Rational series are the notion
\citet{schutzenberger1961definition} introduced for integer coefficients through matrix representations (his
Definition~3), taken here over the nonnegative rationals. The linear systems that sum such a series are the standard
ones of weighted-automata theory \citep{droste2009handbook}. Where a deterministic online presentation exists the
product is an absorbing Markov chain with rewards, and the masses and expectations are the fundamental-matrix quantities
of \citet[Theorems~3.2.4, 3.3.5 and~3.3.7]{kemeny1960finite}. Where none exists, the expected token, boundary and
mismatch counts are expectation-semiring totals in the sense of \citet[Section~4]{eisner2002parameter}, additive rewards
summed over the paths of a weighted automaton. None of the linear algebra is new here, and neither is the raw product.

\citet{benedikt2014unambiguous} read an unambiguous finite automaton against a Markov chain and solve an affine system
of the same form over the trimmed product. Their event is that a trajectory has an accepted prefix, whereas here it is
that the whole string a stopping source emits is accepted. Their solution is unique because every retained state reaches
acceptance with positive probability, and Section~\ref{sec:weighted} runs the same argument, with functionality
supplying the unambiguity and almost sure stopping supplying the geometric bound on the probability of long strings. The
proof of that note's infinite-word theorem was withdrawn, and \citet{baier2023unambiguous} record that its finite-word
lemma is unaffected. \citet{kiefer2019unambiguous} begin with the known memoryless case, in which letters and a stop are
drawn with fixed probabilities and read by an unambiguous automaton. Rows of that product can sum to more than one, and
the multiplicity bound is the run bound any trim unambiguous automaton supplies per prefix.

The positivity law is the support principle for series over the nonnegative rationals, a positive semiring, whose
supports are again recognizable \citep[Theorem~3.1(1)]{kirsten2011support}. That a total is zero exactly when every
coefficient is zero follows from the same positivity. This paper's own part is the identification of that support with
the boundary decision's differing strings inside the source's support, so that the differential mass is zero exactly
when the boundary decision finds no differing string there. The paper also constructs the systems from the supplied
verifiers through the same interface as the decisions.

\paragraph{The hypothesis is a normal form.} Proposition~\ref{prop:collapse} is Eilenberg and Schützenberger's theorem
\citep{eilenberg1974automata} read over the marked alphabet. The general-transducer undecidability of
\citet{filiot2016equivalence} concerns relations that are not functions, and the presentation-level reductions of
Section~\ref{sec:real}, which show that a program promised to compute a policy of the class cannot stand in for a
supplied verifier, are the standard uniformity obstruction.

\paragraph{The frontier, and whose example it is.} The separation of Section~\ref{sec:real} is
\citeauthor{karttunen1996directed}'s example and \citeauthor{choffrut1977sequential}'s theorem, cited there. The same
paper of \citet{karttunen1996directed} builds tokenizers by inserting an end-of-token mark after each match, a marked
string in another encoding, length-changing where the bit encoding used here is length-preserving. The four decisions
ask only for a verifier, which Karttunen's example still has, so they apply beyond the online-scannable region of
Figure~\ref{fig:landscape}.

\paragraph{The instances have finite-state accounts.} \citet{cognetta2025tokenization} constructively represent both BPE
and MaxMatch, the algorithm behind WordPiece, by finite-state transducers, and the extended version of
\citet{berglund2024bpedfa} gives an input-deterministic transducer relating strings to their BPE tokenizations along
with equivalence checking of tokenization dictionaries. Their constructions place canonical BPE inside this paper's
class, which is why Section~\ref{sec:instances} treats it as an exact instance. They also explain why it is not in
general a maximal-munch policy, since BPE, in Cognetta and Okazaki's words, ``does not tokenize strings from left to
right and requires a notion of priority'', its segmentation being determined by an ordered merge list in the sense of
\citet{sennrich2016subword}. How far canonical BPE departs from maximal munch over the same vocabulary varies with the
vocabulary, can be measured as boundary agreement between the two policies, and is not assumed away.

Neither Cognetta and Okazaki nor Berglund and colleagues decide certification or the boundary gap, and Cognetta and
Okazaki mask token choices to preserve an accepted completion, which instantiates the co-accessibility condition of this
paper's realizability inside their own constructions. Both treat composition, Cognetta and Okazaki defining transducer
composition explicitly and Berglund and colleagues naming composition of tokenized languages as an application, so
composition itself is already part of those accounts and this paper's part is composition of segmentation stages with
certificate transfer.

The dictionary equivalence of Berglund and colleagues asks whether two merge dictionaries induce the same canonical
tokenizations, a question about one tokenizer family's outputs, with a distinguishing string available from its automata
by symmetric difference, where the decision here is the full equivalence of arbitrary policies in the class.
\citet{vieira2025canonical} characterize canonical BPE token strings locally, a token string being canonical exactly
when each adjacent pair is, where every vocabulary token is canonical and no pre-tokenizer intervenes. That tests a
token string's own neighbors, where a certificate here is a fixed byte window whose covering token starts at the same
offset in every occurrence.

\paragraph{Computing the scan.} The scan has its own literature, of which \citet{reps1998maximal} is the algorithmic
reference, giving maximal munch in linear time, with the conventions this paper also adopts, earliest definition on a
longest-match tie and no token language containing the empty word.

\section{Limitations}

The results require a finite verifier. Policies with non-regular domains are outside by Lemma~\ref{lem:outside}, a
policy whose domain is the balanced-bracket language among them. Canonical BPE is an instance by the constructions cited
in Section~\ref{sec:instances}, and its distance from maximal munch is a per-vocabulary measurement this paper does not
report. The exact weighted solver is run only on small state spaces, with the certified enclosures handling larger ones,
and scaling the exact solver is open. The lift of Theorem~\ref{thm:lift} is implemented for regular inner stages under
the least-index labeler and for literal outer stages, together with the constant-labeler example, and it is run once
more in Section~\ref{sec:instances} with byte-pair towers as both stages. The literal pipelines are run through the
per-label corollary, which gives a sufficient condition for certificate transfer but does not decide it exactly. The
stack proposition is proved under the full stack's own domain and an innermost certificate, both shown necessary, and
checked by enumeration alone.

Constructing the verifiers limits the scale of the tokenizer instances. The merge tower grows from 4{,}063 states at 64
merges to 69{,}773 at 128 and 841{,}141 at 256, so the dictionaries of Section~\ref{sec:instances}, at 64 to 80 merges,
are nearly three orders of magnitude below a production table such as GPT-2's 50{,}000 merges. For WordPiece the main
experiments use the tabulated verifier at 4{,}096 entries. At 8{,}192 entries the scale probe builds 56{,}636 states and
2{,}965{,}165 transitions in 175 seconds, and a tabulation attempted at 16{,}384 entries did not complete, while the
construction from the vocabulary's tries reaches 16{,}384 entries and the full lexical vocabulary. On the verifiers of
these instances the recorded decisions take seconds to minutes, and the lasso search over the full lexical WordPiece
vocabulary takes 708.1 seconds. The records specify 12 worker processes but do not identify the processor, operating
system, or Python version.

The token automaton of \citet{berglund2024bpedfa} has linearly many states in the dictionary, but the marked-symbol
verifier over bytes derived from it, which these decisions consume, pairs a state of the token automaton, from its state
set $Q_{\mathrm{token}}$, with a node of $\mathrm{Trie}(V)$, the trie of the byte spellings of the vocabulary $V$, so
its size is bounded by the product $O(|Q_{\mathrm{token}}|\,|\mathrm{Trie}(V)|)$ and depends on those explicit
spellings. The trie has up to one node per byte of each token's spelling, and no polynomial in the merge count bounds
it, since a merged token's spelling is the concatenation of the spellings of its two parts, so repeated merges of a
token with itself double its spelling length at each merge. This paper does not run the decisions directly on the token
automaton.

\section{Conclusion}

The paper rests on one hypothesis, that a policy's segmentations are recognizable by a finite verifier over marked
strings. From a supplied verifier it gives constructive decisions for certification, the boundary gap, realizability and
equivalence, three of them specializing classical decisions, and with them two more, whether one certified cut and
whether all the anchors of an inventory let the resulting pieces scan alone. The class is closed under four operations
real pipelines apply, certificates carry across pipeline stages under proved laws, and under rational finite-state byte
sources that stop almost surely the masses of its policies are exact rationals, with the differential mass of two
policies zero exactly when no boundary differential lies in the source's support. The class strictly contains the
finite-state online-scannable policies and in this encoding is exactly the rational policies.

Every executable claim and reported statistic above is one command away from its check, except the incomplete
16{,}384-entry tabulation, which the module documentation records, and the general results rest on the proofs.
Regenerating the library instance's record needs a build of the library, and regenerating the stored records of the
tokenizer instances needs the two published vocabularies. The modules, the gate and the library instance's record are
archived together as a standalone record at
\href{https://doi.org/10.5281/zenodo.23708161}{\mbox{10.5281/zenodo.23708161}}.

\section*{Tools}

Large language models assisted with drafting, with checking citations against primary records and with reviewing the
implementation; the author verified each suggestion by derivation, against primary sources, or against the
implementation, its tests and the archived artifacts, and is responsible for all content.

\setlength{\bibsep}{1pt}
\bibliographystyle{plainnat}
\bibliography{refs}

\end{document}